\documentclass[a4paper,11pt]{scrartcl}

\usepackage{amsmath,amssymb,amsthm}
\usepackage{hyperref, cleveref}
\usepackage{thm-restate}

\usepackage{graphicx}

\usepackage[utf8]{inputenc}
\usepackage[linesnumbered, vlined, ruled]{algorithm2e}
\usepackage{mathrsfs}

\newtheorem{theorem}{Theorem}
\newtheorem{definition}[theorem]{Definition}
\newtheorem{lemma}[theorem]{Lemma}

\newtheorem{corollary}[theorem]{Corollary}

\newcommand{\cost}{\mathsf{cost}}
\newcommand{\prev}{\mathsf{prev}}
\newcommand{\OPT}{\mathcal{O}}
\newcommand{\C}{\mathcal{C}}
\newcommand{\D}{\mathcal{D}}
\newcommand{\A}{\mathcal{A}}
\newcommand{\Ht}{\mathcal{H}}
\newcommand{\Hr}{\mathscr{H}}
\newcommand{\Cr}{\mathscr{C}}

\newcommand{\Bad}{\mathsf{Bad}}
\newcommand{\Anc}{\mathsf{Anc}}
\newcommand{\Ker}{\mathsf{Ker}}

\newcommand{\diam}{\mathsf{diam}}
\newcommand{\rad}{\mathsf{rad}}
\newcommand{\drad}{\mathsf{drad}}
\newcommand{\Aug}{\mathsf{Augment}}
\newcommand{\parent}{\mathsf{parent}}
\newcommand{\nest}{\mathsf{Nest}}

\newcommand{\X}{\mathcal{X}}
\newcommand{\PP}{\mathcal{P}}
\newcommand{\Q}{\mathcal{Q}}

\setkomafont{descriptionlabel}{\normalfont\bfseries}

\title{The Price of Hierarchical Clustering\thanks{This work has been supported by DFG grant RO 5439/1-1}}

\author{Anna Arutyunova\thanks{University of Bonn, Germany, \href{mailto:arutyunova@informatik.uni-bonn.de}{arutyunova@informatik.uni-bonn.de}} \and Heiko R\"oglin\thanks{University of Bonn, Germany, \href{mailto:roeglin@cs.uni-bonn.de}{roeglin@cs.uni-bonn.de}}}
\date{}

\begin{document}

\maketitle

\begin{abstract}
	Hierarchical Clustering is a popular tool for understanding the hereditary properties of a data set. Such a clustering is actually a sequence of clusterings that starts with the trivial clustering in which every data point forms its own cluster and then successively merges two existing clusters until all points are in the same cluster. A hierarchical clustering achieves an approximation factor of~$\alpha$ if  the costs of each $k$-clustering in the hierarchy are at most $\alpha$ times the costs of an optimal $k$-clustering. We study as cost functions the maximum (discrete) radius of any cluster ($k$-center problem) and the maximum diameter of any cluster ($k$-diameter problem).
	
	In general, the optimal clusterings do not form a hierarchy and hence an approximation factor of~$1$ cannot be achieved. We call the smallest approximation factor that can be achieved for any instance the \emph{price of hierarchy}. For the $k$-diameter problem we improve the upper bound on the price of hierarchy to $3+2\sqrt{2}\approx 5.83$. Moreover we significantly improve the lower bounds for $k$-center and $k$-diameter, proving a price of hierarchy of exactly $4$ and $3+2\sqrt{2}$, respectively.
\end{abstract}

\section{Introduction}

Clustering is an ubiquitous task in data analysis and machine learning. In a typical clustering problem, the goal is to partition a set of objects into different clusters such that only similar objects belong to the same cluster. There are numerous ways how clustering can be modeled formally and many different models have been studied in the literature in the last decades. In many theoretical models, one assumes that the data comes from a metric space and that the desired number of clusters is given. Then the goal is to optimize some objective function like $k$-center, $k$-median, or $k$-means. In most cases the resulting optimization problems are NP-hard and hence approximation algorithms have been studied extensively. 

One aspect of real-world clustering problems that is not captured by these models is that it is often already a non-trivial task to determine for a given data set the right or most reasonable number of clusters. One particularly appealing way to take this into account is hierarchical clustering. A hierarchical clustering of a data set is actually a sequence of clusterings, one for each possible number of clusters. It starts with the trivial clustering in which every data point forms its own cluster and then successively merges two existing clusters until all points are in the same cluster. This way for every possible number of clusters, a clustering is obtained. These clusterings help to understand the hereditary properties of the data and they provide information at different levels of granularity.

While hierarchical clustering is successfully used in many applications, it is not as well understood from a theoretical point of view as the models in which the number of clusters is given as part of the input. One reason for this is that it is not obvious how the quality of a hierarchical clustering should be measured. A possibility that has been explored in the literature is to define the quality of a hierarchical clustering based on its worst level. To be precise, let $(\X,d)$ be a metric space and $\PP\subset \X$ a set of $n$ points. Furthermore let $\Hr=(\Ht_n,\ldots,\Ht_1)$ be a hierarchical clustering of $\PP$, where $\Ht_k$ denotes a $k$-clustering, i.e., a clustering with at most $k$ non-empty clusters. Then $\Ht_{k-1}$ arises from $\Ht_k$ by merging some of the existing clusters.
We assume that some objective function like $k$-center, $k$-median, or $k$-means is selected and denote by $\cost(\Ht_k)$ the objective value of $\Ht_k$ with respect to the selected objective function. Furthermore, let $\OPT_k$ denote an optimal $k$-clustering and let $\cost(O_k)$ denote its objective value. Then we say that $\Hr$ achieves an approximation factor of $\alpha\ge 1$ if $\cost(\Ht_k)\le\alpha\cdot\cost(\OPT_k)$ for every $k$, assuming that $\cost$ is an objective that is to be minimized. In this work we consider the radius objective, which is well-known from the $k$-center problem. Here the cost is defined as the maximum radius of a cluster. Furthermore we consider the diameter objective, where the cost is defined as the maximum distance between any two points lying in the same cluster. 

An $\alpha$-approximation for small $\alpha$ yields a strong guarantee for the hierarchical clustering on every level. However, in general there do not exist optimal clusterings $\OPT_n,\ldots,\OPT_1$ that form a hierarchy. So even with unlimited computational resources, a $1$-ap\-pro\-xi\-ma\-ti\-on usually cannot be achieved.
In the literature different algorithms for computing hierarchical clusterings with respect to different objective functions have been developed and analyzed. Dasgupta and Long~\cite{DasguptaL05} and Charikar et al.~\cite{CCFM04} initiated this line of research and presented both independently from each other an algorithm that computes efficiently an 8-approximate hierarchical clustering with respect to the radius and diameter objective. That is, for every level~$k$, the maximal radius or diameter of any cluster in the $k$-clustering computed by their algorithms is at most 8 times the maximal radius or diameter in an optimal $k$-clustering. Inspired by~\cite{DasguptaL05}, Plaxton~\cite{P06} proposed a constant-factor approximation for the $k$-median and $k$-means objective. Later a general framework that also leads constant approximation guarantees for many objective functions including in particular $k$-median and $k$-means has been proposed by Lin et al.~\cite{LNRW10}.

Despite these articles and other related work, which we discuss below in detail, many questions in the area of hierarchical clustering are not yet resolved. We find it particularly intriguing to find out which approximation factors can be achieved for different objectives. This question comes in two flavors depending on the computational resources available. Of course it is interesting to study which approximation factors can and cannot be achieved in polynomial time, assuming P~$\neq$~NP. Since in general there do not exist hierarchical clusterings that are optimal on each level, it is also interesting to study which approximation factors can and cannot be achieved in general without the restriction to polynomial-time algorithms.

For an objective function like radius or diameter we define its \emph{price of hierarchy} as the smallest $\alpha$ such that for any instance there exists an $\alpha$-approximate hierarchical clustering. Hence, the price of hierarchy is a measure for how much quality one has to sacrifice for the hierarchical structure of the clusterings. 

Our main results are tight bounds for the price of hierarchy for the radius, discrete radius and diameter objective. Here the difference between radius and discrete radius lies in the choice of centers. For the radius objective we allow to choose the center of a cluster $C\subset \PP$ from the whole metric space $\X$, while for the discrete radius objective the center must be contained in $C$ itself. We will see that this has an impact on the price of hierarchy. For all three objectives the algorithms in~\cite{DasguptaL05,CCFM04} compute an $8$-approximate hierarchical clustering in polynomial time. Until recently this was also the best known upper bound for the price of hierarchy in the literature for hierarchical radius and diameter. For discrete radius, Gro{\ss}wendt~\cite{G20} shows an upper bound for the price of hierarchy of~4. The best known lower bounds are 2, proven by Das and Kenyon-Mathieu~\cite{DM09} for diameter and by Gro{\ss}wendt~\cite{G20} for (discrete) radius. We improve the framework in~\cite{LNRW10} for radius and diameter and show an upper bound on the price of hierarchy of $3+2\sqrt{2}\approx 5.83$. The upper bound of $3+2\sqrt{2}$ for the radius was also recently proved by Bock~\cite{Bock2022} in independent work. 
However our main contribution lies in the design of clustering instances to prove a lower bound of~4 for discrete radius and $3+2\sqrt{2}$ for radius and diameter.

\subparagraph*{Related work.}
Gonzales~\cite{G85} presents a simple and elegant incremental algorithm for $k$-center. The algorithm exhibits the following nice property: given a set $\PP$ which has to be clustered, it returns an ordering of the points, such that the first $k$ points constitute the centers of the $k$-center solution, and this solution is a $2$-approximation for every  $1\leq k\leq |\PP|$. However the resulting clusterings are usually not hierarchically compatible. Dasgupta and Long~\cite{DasguptaL05} use the ordering computed by Gonzales' algorithm to compute a hierarchical clustering. The authors present an 8-approximation for the objective functions (discrete) radius and diameter. In an independent work Charikar et al.~\cite{CCFM04} also present an $8$-approximation for the three objectives which outputs the same clustering as the algorithm in~\cite{DasguptaL05} under some reasonable conditions~\cite{DM09}. In a recent work, Mondal~\cite{Mdl18} gives a $6$-approximation for hierarchical (discrete) radius. In Appendix~\ref{ap:mondal} we present an instance where this algorithm computes only a $7$-approximation contradicting the claimed guarantee.

Plaxton~\cite{P06} shows that a similar approach as in~\cite{DasguptaL05} yields a hierarchical clustering with constant approximation guarantee for the $k$-median and $k$-means objectives. Later a general framework for a variety of incremental and hierarchical problems was introduced by Lin et al.~\cite{LNRW10}. Their framework can be applied to compute hierarchical clusterings for any cost function which satisfies a certain nesting property, especially those of $k$-median and $k$-means. This yields a $20.71\alpha$-approximation for $k$-median and a $576\beta$-approximation for $k$-means. Here $\alpha=2.675$ and $\beta=6.357$ are the currently best approximation guarantees for $k$-median~\cite{BPRST17} and $k$-means~\cite{ANSW17}.
The algorithms presented in \cite{CCFM04,DasguptaL05,LNRW10,P06} run in polynomial time. Unless P=NP there is no polynomial-time $\alpha$-approximation for $\alpha<2$ for hierarchical (discrete) radius and diameter. For (discrete) radius this is an immediate consequence of the reduction from dominating set presented by~\cite{HS86}. A similar reduction from clique cover yields the statement for hierarchical diameter. 

However even without time constraints it is not clear what approximation guarantee can be achieved for hierarchical clustering. It is easy to find examples, where the approximation guarantee of any hierarchical clustering for all three objectives is greater than one. Das and Kenyon-Mathieu \cite{DM09} and Gro\ss wendt~\cite{G20} present instances for diameter and (discrete) radius, where no hierarchical clustering has an approximation guarantee smaller than $2$.
On the other hand Gro\ss wendt~\cite{G20} proves an upper bound of $4$ on the approximation guarantee of hierarchical discrete radius by using the framework of Lin et al.~\cite{LNRW10}. In recent independent work Bock~\cite{Bock2022} improved the bound for hierarchical radius to $3+2\sqrt{2}$. While his approach is inspired by Dasgupta and Long~\cite{DasguptaL05}, the resulting algorithm is similar to the algorithm we present  in this paper as an improvement of~\cite{LNRW10}. 

Aside from the theoretical results, there also exist greedy heuristics, which are more commonly used in applications. One very simple bottom up, also called agglomerative, algorithm is the following: starting from the clustering where every point is separate, it merges in every step the two clusters whose merge results in the smallest increase of the cost function. For (discrete) radius and diameter this algorithm is known as complete linkage and for the $k$-means cost this is Ward's method~\cite{W63}.  Ackermann et al.~\cite{ABKS14} analyze the approximation guarantee of complete linkage in the Euclidean space. They show an approximation guarantee of $O(\log(k))$ for all three objectives assuming the dimension of the Euclidean space to be constant. This was later improved by Gro\ss wendt and R\"oglin~\cite{GR17} to $O(1)$. In arbitrary metric spaces complete linkage does not perform well. There Arutyunova et al.~\cite{AGRSW21} prove a lower bound of $\Omega(k)$ for all three objectives. For Ward's method Gro\ss wendt et al.~\cite{GRS19} show an approximation guarantee of $2$ under the strong assumption that the optimal clusters are well separated.  

Recently other cost functions for hierarchical clustering were proposed, which do not compare to the optimal clustering on every level. Dasgupta~\cite{D16} defines a new cost function for similarity measures and presents an $O(\alpha\log(n))$-approximation for the respective problem. This was later improved to $O(\alpha)$ independently by Charikar and Chatziafratis~\cite{CC17} and Cohen-Addad et al.~\cite{CKMM18}. Here $\alpha$ is the approximation guarantee of sparsest cut. However Cohen-Addad et al.~\cite{CKMM18} prove that every hierarchical clustering is an $O(1)$-approximation to the corresponding cost function for dissimilarity measures when the dissimilarity measure is a metric. A cost function more suitable for Euclidean spaces was developed by Wang and Moseley~\cite{WM20}. They prove that a randomly generated hierarchical clustering performs poorly for this cost function and show that bisecting $k$-means computes an $O(1)$-approximation.

\subparagraph*{Our results.} 
We define the \emph{price of hierarchy} $\rho_{\cost}$ with respect to an objective function $\cost$ as the smallest number such that for every clustering instance there exists a hierarchical clustering which is a $\rho_{\cost}$-approximation with respect to $\cost$. Observe that the results~\cite{CCFM04,DM09,DasguptaL05,G20} imply that the price of hierarchy for radius and diameter is between $2$ and $8$ and for discrete radius between $2$ and $4$. We close these gaps and prove that the price of hierarchy for radius and diameter is exactly $3+2\sqrt{2}$ and for discrete radius exactly $4$. Notice that this does not imply the existence of polynomial-time algorithms with approximation guarantee $\rho_\cost$. Especially our algorithm which computes a $3+2\sqrt{2}$-approximation for radius and diameter does not run in polynomial time. This is also the case for the $3+2\sqrt{2}$-approximation for radius presented by Bock~\cite{Bock2022} in independent work. Our upper bound of $3+2\sqrt{2}$ can be achieved by a small improvement in the framework of Lin et al.~\cite{LNRW10}. However our most technically demanding contribution is the design of a clustering instance for every $\epsilon>0$ such that every hierarchical clustering has approximation guarantee at least $ 3+2\sqrt{2}-\epsilon$ for radius and diameter and $4-\epsilon$ for discrete radius. It requires a careful analysis of all possible hierarchical clusterings, which is highly non-trivial for complex clustering instances.  

\section{Preliminaries}
A clustering instance $(\X,\PP,d)$ consists of a metric space $(\X,d)$ and a finite subset $\PP\subset \X$. For a set (or cluster) $C\subset \PP$ we denote by
\[\diam(C)=\max_{p,q\in C}d(p,q)\]
the \emph{diameter} of $C$. By $\rad(C,c)=\max_{p\in C}d(c,p)$ we denote the radius of $C$ with respect to a center $c\in \X$. This is the largest distance between $c$ and a point in $C$.
The \emph{radius} of $C$ is defined as the smallest radius of $C$ with respect to a center $c\in \X$, i.e.,
\[\rad(C)=\min_{c\in \X}\rad(C,c)\]
while the \emph{discrete radius} of $C$ is defined as the smallest radius of $C$ with respect to a center $c\in C$, i.e.,
\[\drad(C)=\min_{c\in C}\rad(C,c).\]

A $k$-clustering of $\PP$ is a partition of $\PP$ into at most $k$ non-empty subsets. We consider three closely related clustering problems.

The $k$-diameter problem asks to minimize the maximum diameter $\diam(\C_k)=$ \linebreak $\max_{C\in\C_k}\diam(C)$ of a $k$-clustering $\C_k$.
In the $k$-center problem we want to minimize the maximum radius
$\rad(\C_k)=\max_{C\in \C_k}\rad(C)$, 
and in the discrete $k$-center problem we want to minimize the maximum discrete radius
$\drad(\C_k)=\max_{C\in \C_k}\drad(C)$.

\begin{definition}
	Given an instance $(\X,\PP,d)$, let $n=|\PP|$. We call two clusterings $\C$ and $\C'$ of $\PP$ with $|\C|\geq |\C'|$ hierarchically compatible if for all $C\in \C$ there exists $C'\in\C'$ with $C\subset C'$.  
	A \emph{hierarchical clustering} of $\PP$ is a sequence of clusterings $\Hr=(\Ht_n, \ldots, \Ht_1)$, such that 
	\begin{enumerate}
		\item $\Ht_i$ is an $i$-clustering of $\PP$
		\item for $1< i\leq n$ the two clusterings $\Ht_{i-1}$ and $\Ht_{i}$ are hierarchically compatible.
	\end{enumerate}
	For $\cost \in \{\diam,\rad,\drad\}$ let $\OPT_i$ denote the optimal $i$-clustering with respect to $\cost$. We say that $\Hr$ is an $\alpha$-approximation with respect to $\cost$ if for all $i=1,\ldots, n$ we have 
	\[\cost(\Ht_i)\leq \alpha\cdot\cost(\OPT_i).\]
\end{definition}
Since optimal clusterings are generally not hierarchically compatible, there is usually no hierarchical clustering with approximation guarantee $\alpha=1$. We have to accept that the restriction on hierarchically compatible clusterings comes with an unavoidable increase in the cost compared to an optimal solution.
\begin{definition}
	For $\cost\in\{\diam,\rad,\drad\}$ the \emph{price of hierarchy} $\rho_\cost\geq 1$ is defined as follows. 
	\begin{enumerate}
		\item For every instance $(\X,\PP,d)$, there exists a hierarchical clustering $\Hr$ of $\PP$ that is a $\rho_{\cost}$-approximation with respect to $\cost.$
		\item For any $\alpha<\rho_\cost$ there exists an instance $(\X,\PP,d)$, such that there is no hierarchical clustering of $\PP$ that is an $\alpha$-approximation with respect to $\cost$.
	\end{enumerate} 
\end{definition}  
Thus $\rho_\cost$ is the smallest possible number such that for every clustering instance there is a hierarchical clustering with approximation guarantee $\rho_\cost$. 

\section{An Upper Bound on the Price of Hierarchy}
\label{chap:upper_bound}
The framework by by Lin et al.~\cite{LNRW10} can be applied to compute incremental and hierarchical solutions to a large class of minimization problems.
We already know that the framework yields an upper bound of $4$ on the price of hierarchy for the discrete radius~\cite{G20}.
It also yields upper bounds for the price of hierarchy for radius and diameter, which are not tight, however. We first discuss the framework in the context of hierarchical clustering for (discrete) radius and diameter. In the second part we then present an improved version of their algorithm for radius and diameter.

\begin{restatable}{theorem}{ThmUpperBound}
	\label{thm:upper_bound}
	For $\cost\in\{\diam,\rad\}$ we have $\rho_\cost\leq 3+2\sqrt{2}\approx 5.828$.
\end{restatable}
First we introduce the notion of a \emph{hierarchical sequence}, which is a relaxation of a hierarchical clustering in the sense that it does not have to contain a $k$-clustering for every $1\leq k\leq |\PP|$.
\begin{definition}
	\label{def:hierarical_ext}
	Given an instance $(\X,\PP,d)$, with $n=|\PP|$. We call a sequence $\Cr=(\C^{(t)}, \ldots, \C^{(1)})$ of clusterings a hierarchical sequence if it satisfies
\begin{enumerate}
	\item $|\C^{(t)}|=n$ and $|\C^{(1)}|=1$
	\item for $1\leq i\leq t$ either $\C^{(i-1)}=\C^{(i)}$ or $\C^{(i-1)}$ is obtained from $\C^{(i)}$ by merging some of its clusters.
\end{enumerate}
Such a hierarchical sequence can be extended to a hierarchical clustering of $\PP$ as follows. 
We define the respective hierarchical clustering $h(\Cr)$ by assigning every $1\leq i\leq n$ the clustering among $\C^{(t)},\ldots, \C^{(1)}$ of smallest cost and size at most $i$. We say that $\Cr$ is an $\alpha$-approximation iff $h(\Cr)$ is an $\alpha$-approximation.
\end{definition}

Before we are able to define the algorithm we need one important definition from~\cite{LNRW10}.
\begin{definition}
	Given an instance $(\X,\PP,d)$. For $\cost\in\{\diam, \rad,\drad\}$ we say that the $(\gamma, \delta)$-nesting property holds for reals $\gamma, \delta\geq 0$, if for any two clusterings $\C, \D$ of $\PP$ with $|\C|>|\D|$ there exists a clustering $\C'$ with 
	\begin{enumerate}
		\item $|\C'|\leq|\D|$
		\item $\C'$ is hierarchically compatible with $\C$ and
		\item $\cost(\C')\leq \gamma \cost(\C)+\delta\cost(\D)$.
	\end{enumerate}
	We say that $\C'$ is a \emph{nesting} of $\C$ at $\D$. Let $\Aug_{\cost}(\C,\D,\gamma, \delta)$ denote the subroutine that computes such a clustering $\C'$.
\end{definition}

\begin{algorithm}[t]
	\SetKwInOut{Input}{Input}
	\SetKwInOut{Output}{Output}
	\Input{Clustering instance $(\X,\PP,d)$, with $d(x,y)>2$ for all $x,y\in \PP$, optimal clusterings $\OPT_{|\PP|},\ldots, \OPT_1$ of $\PP$ with respect to $\cost$}
	\Output{A hierarchical clustering of $\PP$}
	\DontPrintSemicolon
	\SetKwFor{ForAll}{for all}{}{}
	\BlankLine
	Set $\Delta=\cost(\OPT_1), t=\lceil\log_{2\gamma}(\Delta)\rceil+1$ and $\C^{(t)}=\OPT_{|\PP|}$\;
	\For{$i=t-1$ \KwTo $1$}
	{
		Let $1\leq n_i\leq |\PP|$ be the smallest number such that $\cost(\OPT_{n_i})\in((2\gamma)^{t-i-1}, (2\gamma)^{t-i}]$\;
		\If{\textup{such a number exists}}{
			set $\C^{(i)}=\Aug_\cost(\C^{(i+1)}, \OPT_{n_i}, \gamma, \delta)$\;
		}
		\Else{set $\C^{(i)}=\C^{(i+1)}$\;}
	}
	\Return $h((\C^{(t)},\ldots, \C^{(1)}))$\; 
	\caption{(Lin et al.~\cite{LNRW10})}
	\label{Alg:Nest}
\end{algorithm}
The algorithm of Lin et al.~\cite{LNRW10} is shown as Algorithm~\ref{Alg:Nest}. It computes a hierarchical sequence $\Cr=(\C^{(t)},\ldots, \C^{(1)})$ of clusterings as follows. Starting with $\C^{(t)}=\OPT_{|\PP|}$ the algorithm builds the $i$-th clustering $\C^{(i)}$ as nesting of $\C^{(i+1)}$ at an optimal clustering $\OPT_{n_i}$. This guarantees that the clusterings are hierarchically compatible. 
\begin{theorem}[\cite{LNRW10}]
	\label{thm:Lin}
	For $\cost\in\{\drad,\rad,\diam\}$, if the $(\gamma, \delta)$-nesting property holds for reals $\gamma\geq 1, \delta>0$ then Algorithm~\ref{Alg:Nest} computes a hierarchical clustering of $\PP$ with approximation guarantee $4\gamma\delta$.
\end{theorem}
Gro\ss wendt~\cite{G20} proved the existence of such a nesting property for $\diam, \rad$, and $\drad$.
\begin{lemma}[\cite{G20}]
	\label{lemma:Gross}
		For $\cost\in\{\diam,\rad\}$ there exists a $(2,1)$-nesting and for $\cost=\drad$ there exists a $(1,1)$-nesting.  
\end{lemma}
In combination with Theorem~\ref{thm:Lin} this yields $\rho_{\drad}\leq 4$. However, for the other two objectives we obtain an upper bound of only $8$. We improve Algorithm~\ref{Alg:Nest} to obtain the claimed upper bound of $3+2\sqrt{2}$. 

In the definition of the $(\gamma, \delta)$-nesting property we require a nesting of $\C$ at $\D$ for arbitrary clusterings $\C,\D$ with $|\C|>|\D|$. However, in Algorithm~\ref{Alg:Nest} we know more about the structure of $\C$. This clustering is obtained by repeatedly nesting at optimal clusterings of increasing cost. In Algorithm~\ref{Alg:NestNew} we define a nesting subroutine for this type of clusterings that eventually leads to a better approximation-guarantee.

\begin{algorithm}[t]
	\SetKwInOut{Input}{Input}
	\SetKwInOut{Output}{Output}
	\DontPrintSemicolon
	\SetKwFor{ForAll}{for all}{}{}
	\BlankLine
	
	\Input{Step size $\alpha>1$. Clustering instance $(\X,\PP,d)$, with $d(x,y)>2$ for all $x,y\in \PP$, optimal clusterings $\OPT_{|\PP|},\ldots, \OPT_1$ of $\PP$ with respect to $\cost$}
	\Output{A hierarchical clustering of $\PP$}
	
	Set $\Delta=\cost(\OPT_1), t=\lceil\log_{\alpha}(\Delta)\rceil+1$ and  $\C^{(t)}=\OPT_{|\PP|}$\;
	For all $C\in \C^{(t)}$ we set $\parent_t(C)=C$\;
	\For{$i=t-1$ \KwTo $1$}
	{
		Let $1\leq n_i\leq |\PP|$ be the smallest number such that $\cost(\OPT_{n_i})\in(\alpha^{t-i-1}, \alpha^{t-i}]$\;
		\If{\textup{such a number exists}}{
		For $C\in \C^{(i+1)}$ let $O\in \OPT_{n_i}$ be a cluster with $\parent_{i+1}(C)\cap O\neq \emptyset$ and set $\nest_i(C)=O$\;
		Set $\C^{(i)}=\{\bigcup_{C\in \nest_i^{(-1)}(O)}C\mid O\in\OPT_{n_i}\}$\;
		Set $\parent_i(\bigcup_{C\in Nest_i^{(-1)}(O)}C)=O$ for all $O\in\OPT_{n_i}$\;
		}
		\Else{set $\C^{(i)}=\C^{(i+1)}, \parent_{i}=\parent_{i+1}$\;}
	}
	\Return $h((\C^{(t)},\ldots, \C^{(1)}))$\; 	
	\caption{}
	\label{Alg:NestNew}
\end{algorithm}
The main difference between Algorithm~\ref{Alg:Nest} and Algorithm~\ref{Alg:NestNew} is the replacement of the function $\Aug_\cost(\C_{i+1}, \OPT_{n_i},\gamma,\delta)$, which computes the nesting of $\C^{(i+1)}$ at $\OPT_{n_i}$, by a more explicit approach to compute such a nesting. 
We use the fact that $\C^{(i+1)}$ is obtained by a nesting at $\OPT_{n_{i+1}}$. This is reflected in the function $\parent_{i+1}$ which assigns every cluster in $C^{(i+1)}$ a cluster from $\OPT_{n_{i+1}}$. 
In iteration $i$ we then use the $(i+1)$-st parent function to determine which clusters of $\C^{(i+1)}$ will be merged to obtain $\C^{(i)}$. We are allowed to merge clusters $C,D\in\C^{(i+1)}$ if there is a cluster $O\in \OPT_{n_i}$ which has a non-empty intersection with both, $\parent_{i+1}(C)$ and $\parent_{i+1}(D)$. The parent of the merged cluster in $\C^{(i)}$ is then set to $O$.
\begin{restatable}{lemma}{LemmaUpperBound}
	\label{lemma:upper_bound}
	For $\cost\in\{\diam,\rad\}$ and any $\alpha>1$ Algorithm~\ref{Alg:NestNew} computes a hierarchical clustering with approximation guarantee $\alpha\big(\frac{2}{\alpha-1}+1\big)$.
\end{restatable}
\begin{proof} Let $n$ denote the cardinality of $\PP$.
	Notice first that $(\C^{(t)},\ldots, \C^{(1)})$ is indeed a hierarchical sequence. The first property of a hierarchical sequence is satisfied: We define $\C^{(t)}=\OPT_{n}$ and since $\cost(\OPT_1)=\Delta \in (\alpha^{t-2},\alpha^{t-1}]$ we obtain $|\C^{(1)}|\leq n_1=1$. The second property is satisfied since $\C^{(i)}$ either equals $\C^{(i+1)}$ or is obtained by merging clusters from $\C^{(i+1)}$.
	Thus Algorithm~\ref{Alg:NestNew} indeed computes a hierarchical clustering.
	
	\textbf{Diameter ($\cost=\diam$):}
	Let $1\leq i\leq t$. We claim
	\begin{enumerate}
		\item for every cluster $C\in \C^{(i)}$ and every point $p\in \parent_i(C)$ that
		$\max_{q\in C}d(p,q)\leq \sum_{l=1}^{t-i}\alpha^l,$
		\item that $\diam(\C^{(i)})\leq \alpha^{t-i}+2\sum_{l=1}^{t-i-1}\alpha^l$.
	\end{enumerate} 
	We prove this by induction over $i$, starting with $i=t$ in decreasing order. Observe that $\C^{(t)}$ consists only of clusters of size one so these claims are true for $i=t$.
	
	Let $1\leq i\leq t-1$. If $\C^{(i)}=\C^{(i+1)}$ both claims are true by induction hypothesis. Thus we assume from now on that $\C^{(i)}\neq\C^{(i+1)}$. For the first claim, we fix a cluster $C\in C^{(i)}$ and two points $p\in \parent_i(C)$ and $q\in C$. Let $D\in \C^{(i+1)}$ be the cluster which contains $q$. Since $\C^{(i)}$ is obtained by merging clusters from $\C^{(i+1)}$, we know that $D\subset C$ and thus $\parent_{i+1}(D)\cap \parent_i(C)\neq \emptyset.$ Let $x\in \parent_{i+1}(D)\cap \parent_i(C)$. By the induction hypothesis 
	\[d(x,q)\leq \max_{y\in D}d(x,y)\leq\sum_{l=1}^{t-i-1}\alpha^l.\]
	Since $p$ and $x$ lie both in $\parent_i(C)$ we obtain $d(p,x)\leq \diam(\OPT_{n_i})\leq \alpha^{t-i}$. Using the triangle inequality we conclude
	\[d(p,q)\leq d(p,x)+d(x,q)\leq \sum_{l=1}^{t-i}\alpha^l.\]
	
	For the second claim we again fix a cluster $C\in\C^{(i)}$ and two points $p,q\in C$. Let $B,D\in \C^{(i+1)}$ such that $p\in B$ and $q\in D$. Observe that $B\cup D\subset C$ and thus $\parent_{i+1}(B)\cap \parent_i(C)\neq\emptyset\neq \parent_{i+1}(D)\cap \parent_i(C)$. Let $x_p\in \parent_{i+1}(B)\cap \parent_i(C)$ and $x_q\in \parent_{i+1}(D)\cap \parent_i(C)$. 
	Since $x_p$ and $x_q$ lie both in $\parent_i(C)$ we obtain $d(x_p,x_q)\leq \diam(\OPT_{n_i})\leq \alpha^{t-i}.$
	We apply the triangle inequality and the induction hypothesis to obtain
	\[d(p,q)\leq d(p,x_p)+d(x_p,x_q)+d(x_q,q)\leq \alpha^{t-i}+2\sum_{l=1}^{t-i-1}\alpha^l.\] 
	
	\textbf{Radius ($\cost=\rad$):}
	Let $1\leq i\leq t$. We claim that for every cluster $C\in \C^{(i)}$ and the center $c$ of cluster $\parent_i(C)$ holds $\max_{q\in C}d(c,q)\leq \alpha^{t-i}+2\sum_{l=1}^{t-i-1}\alpha^l$. Notice that this immediately implies \[\rad(\C^{(i)})\leq\alpha^{t-i}+2\sum_{l=1}^{t-i-1}\alpha^l.\]
	
	We prove this by induction over $i$. Observe that $\C^{(t)}$ consists only of clusters of size one. So this claim is true for $i=t$. Let $1\leq i\leq t-1$. If $\C^{(i)}=\C^{(i+1)}$ the claim is true by induction hypothesis. Thus we assume from now on that $\C^{(i)}\neq\C^{(i+1)}$. We fix a cluster $C\in C^{(i)}$ a point $q\in C$ and denote by $c$ the center of $\parent_i(C)$. Let $D\in \C^{(i+1)}$ be the cluster which contains $q$. Since $\C^{(i)}$ is obtained by merging clusters from $\C^{(i+1)}$, we know that $D\subset C$ and thus $\parent_{i+1}(D)\cap \parent_i(C)\neq \emptyset.$ Let $x\in \parent_{i+1}(D)\cap \parent_i(C)$. By induction hypothesis the following holds for the center $d$ of $\parent_{i+1}(D)$ 
	\[\max_{v\in D}d(d,v)\leq \alpha^{t-i-1}+2\sum_{l=1}^{t-i-2}\alpha^l.\]
	Together with the triangle inequality this implies
	\[d(x,q)\leq d(x,d)+d(d,q)\leq \rad(\OPT_{n_{i+1}})+\alpha^{t-i-1}+2\sum_{l=1}^{t-i-2}\alpha^l\leq 2\sum_{l=1}^{t-i-1}\alpha^l.\]
	This yields the claim, as 
	\[d(c,q)\leq d(c,x)+d(x,q)\leq \rad(\OPT_{n_i})+2\sum_{l=1}^{t-i-1}\alpha^l\leq \alpha^{t-i}+2\sum_{l=1}^{t-i-1}\alpha^l.\]
	
	Finally we can bound the approximation factor for both radius and diameter. Let $\cost\in\{\diam, \rad\}$.  Since $d(x,y)>2$ for all $x,y\in \PP$ we get that $\cost(\OPT_{n-1})>1$. Thus for every $1\leq m<n$ there is $1\leq i\leq t-1$ such that $\cost(\OPT_m)\in (\alpha^{t-i-1},\alpha^{t-i}]$. 
	Thus the clustering $h((\C^{(t)},\ldots,\C^{(1)}))$ is an $\alpha\Big(\frac{2}{\alpha-1}+1\Big)$-approximation iff for all $1\leq i\leq t$ 
	\[\cost(\C^{(i)})\leq \alpha\Big (\frac{2}{\alpha-1}+1\Big )\cost(\OPT)\] for all optimal clusterings $\OPT$ with $\cost(\OPT)\in(\alpha^{t-i-1},\alpha^{t-i}]$. We obtain
	\[\cost(\C^{(i)})\leq \alpha^{t-i}+2\sum_{l=1}^{t-i-1}\alpha^l< \alpha^{t-i}+2\cdot\frac{\alpha^{t-i}}{\alpha-1}=\alpha^{t-i}\Big (\frac{2}{\alpha-1}+1\Big )\leq \alpha\Big (\frac{2}{\alpha-1}+1\Big)\cost(\OPT).\]
\end{proof}
\ThmUpperBound*
\begin{proof}
	Let $(\X,\PP,d)$ be a clustering instance. We can assume without loss of generality that $d(x,y)>2$ for all $x,y\in \PP$, otherwise we scale the metric $d$ accordingly. 
	For $\cost\in\{\diam,\rad\}$ we then use Algorithm~\ref{Alg:NestNew} with $\alpha=1+\sqrt{2}$ to compute a hierarchical clustering.  
	By Lemma~\ref{lemma:upper_bound} we obtain a hierarchical clustering that is an $3+2\sqrt{2}$ approximation and thus $\rho_\cost\leq 3+2\sqrt{2}$.
\end{proof}

\section{A Lower Bound on the Price of Hierarchy}
\label{chap:lower_bound}
The most challenging contributions of this article are matching lower bounds on the price of hierarchy for diameter, radius, and discrete radius. 
\begin{restatable}{theorem}{ThmLowerBound}
	\label{thm:lower_bound}
	For $\cost\in\{\diam,\rad\}$ we have $\rho_\cost\geq 3+2\sqrt{2}$ and for $\cost=\drad$ we have $\rho_\cost\geq 4$.
\end{restatable}    
There is already existing work in this area by Das and Kenyon-Mathieu~\cite{DM09} for the diameter and Gro\ss wendt~\cite{G20} for the radius. Both show a lower bound of $2$ for the respective objective. To improve upon these results we have to construct much more complex instances which differ significantly from those in \cite{DM09,G20}.

For every $\epsilon>0$ we will construct a clustering instance $(\X,\PP,d)$ such that for any hierarchical clustering $\Hr=(\Ht_{|\PP|},\ldots, \Ht_1)$ of $\PP$ there is $1\leq i\leq |\PP|$ such that $\cost(\Ht_i)\geq \alpha\cdot\cost(\OPT_i)$, where $\OPT_i$ is an optimal $i$-clustering of $\PP$ with respect to $\cost$ and $\alpha=(3+2\sqrt{2}-\epsilon)$ for $\cost\in\{\diam,\rad\}$ and $\alpha=4-\epsilon$ for $\cost=\drad$.

The proof is divided in three parts. First we introduce the clustering instance $(\X,\PP,d)$ and determine its optimal clusterings. In the second part we develop the notion of a \emph{bad} cluster. We prove that any hierarchical clustering contains such bad clusters and develop a lower bound on their cost. In the third part we compare the lower bound to the cost of optimal clusterings and prove Theorem~\ref{thm:lower_bound}. 

\subsection{Definition of the Clustering Instance}
For $n\in\mathbb N$ we denote by $[n]$ the set of numbers from $1$ to $n$. 

Let $k\in\mathbb N$ and $\Gamma=k+1$. For $0\leq \ell\leq k$ we define point sets $\Q_\ell$ and $\PP_\ell$ as follows
\begin{enumerate}
	\item For $\ell=0$ let $\PP_0=\Q_0=[1]$ and denote by $N_0$ the cardinality of $\PP_0$.
	\item For $\ell>0$ let $\Q_{\ell}=[\Gamma\cdot N_{\ell-1}]^{N_{\ell-1}}$ and $\PP_\ell=\prod_{i=0}^{\ell} \Q_i$. Furthermore set  $N_\ell=|\PP_\ell|$.
\end{enumerate}
Moreover let $\phi_\ell\colon \PP_\ell\rightarrow \big[N_\ell]$ be a bijection for $0\leq \ell\leq k$.

We refer to a point $X\in \PP_k$ as a matrix with $k+1$ rows and $N_{\ell-1}$ entries in the $\ell$-th row. Thus we write
\[X=(x_{01}\mid\ldots\mid x_{\ell1},\ldots, x_{\ell N_{\ell-1}}\mid\ldots\mid x_{k1},\ldots, x_{kN_{k-1}}).\] Let $X_{\ell}=(x_{\ell 1},\ldots, x_{\ell N_{\ell-1}})\in \Q_{\ell}$ for $0\leq \ell\leq k$. For a shorter representation we can replace the $\ell$-th row directly by $X_\ell$ and for $0\leq i\leq j\leq k $ we can replace the $i$-th up to $j$-th row by $X_{[i:j]}=(X_i\mid\ldots\mid X_j)$.
 
Let $X\in \PP_k$ and $ 1\leq \ell\leq k$. Notice that $X_{[0:\ell-1]}\in\PP_{\ell-1}$ and let $m=\phi_{\ell-1}(X_{[0:\ell-1]})$, we define 
\begin{align*}
A^X_\ell=\{(X_{[0: \ell-1]}\mid x_{\ell1},\ldots,x_{\ell m-1},\star, x_{\ell m+1}, \ldots, x_{\ell N_{\ell-1}}\mid X_{[\ell+1: k]})\mid \star\in[\Gamma \cdot N_{\ell-1}]\}.
\end{align*}
Thus all coordinates of points in $A^X_{\ell}$ are fixed and agree with those of $X$ except one which is variable. Here $X_{[0:\ell-1]}$ serves as prefix which indicates through $\phi_{\ell-1}$ which coordinate of $X_\ell$ can be changed. 

We define $\A_{\ell}=\{A_\ell^X\mid X\in \PP_{k}\}$ as the set containing all subsets of this form. It is clear that $\A_\ell$ is a partition of $\PP_k$ and that it contains only sets of size $\Gamma \cdot N_{\ell-1}$. Furthermore we set $\A_0=\{\{X\}\mid X\in \PP_k\}$. 

Let $G=(V, E, w)$ denote the weighted hyper-graph with $V=\PP_k$ and $E=\bigcup_{i=1}^k\A_i$. The weight of a hyper-edge $e \in E$ is set to $\ell$ iff $e\in\A_\ell$. For $0\leq \ell\leq k$, the sub-graph $G_{\ell}=(V_\ell, E_\ell, w_{\ell})$ is given by $V_{\ell}=\PP_k, E_{\ell}=\bigcup_{i=0}^\ell \A_i$ and $w_{\ell}=w_{|E_\ell}$.

We extend $G$ to a hyper-graph $H=(V',E',w')$ as follows. Let $V'=V\cup \bigcup_{i=0}^{k}\{v_A\mid A\in\A_i\}$ and $E'=E\cup\bigcup_{i=0}^{k}\{\{v,v_A\}\mid A\in\A_i, v\in A\}$. Thus $H$ contains one vertex for every $A\in \bigcup_{i=0}^k \A_i$ and this vertex is connected by edges to every vertex $v\in A$. For $e\in E$ we set $w'(e)=w(e)$ and for $e=\{v,v_A\}$ for some $A\in \A_\ell$ and $v\in A$ we set $w'(e)=\ell/2$.  

The clustering instance $(\X,\PP,d)$ is given by $\X=V', \PP=V$, and $d$ as the shortest path metric on $H$. Observe that the extension of~$G$ to~$H$ is only necessary for the lower bound for the radius but not for the diameter and the discrete radius. This is because the additional points~$V'\setminus V$ do not belong to $\PP$ and are hence irrelevant for the clustering instance for the diameter and discrete radius. In the lower bound for the radius they will be used as centers, however.

\begin{lemma}
	\label{lemma_metric}
	Let $p,q\in V$, then $d(p,q)$ is the length of a shortest path between $p$ and $q$ in $G$.
\end{lemma}
\begin{proof}
	By definition $d(p,q)$ is the length of a shortest path between $p$ and $q$ in $H$. Suppose the shortest path contains a vertex $v_A$ for some $A\in\bigcup_{i=0}^k \A_i$ with $v\in A$ as predecessor and $w\in A$ as ancestor. Since $v$ and $w$ are connected in $H$ by the hyper-edge $A$ we can delete $v_A$ from the path and the length of the path does not change. The resulting path is also a path in $G$, so $d(p,q)$ is also the length of a shortest path between $p$ and $q$ in $G$. 
\end{proof}

Next we state some structural properties of the graph $G$ and the clustering instance $(\X,\PP,d)$. 
To establish a lower bound on the approximation factor of a hierarchical clustering we first focus on the optimal clusterings of the instance $(\X,\PP,d)$. One can already guess that  $\A_\ell$ is an optimal clustering with $\frac{N_k}{\Gamma N_{\ell-1}}$ clusters with respect to $\cost\in\{\diam,\rad,\drad\}$ and we will prove this in this section. First we need the following statement about the connected components of $G_\ell$.
\begin{lemma}
	\label{lem2:construction_components}
	The vertex set of every connected component in $G_\ell$ has cardinality $N_{\ell}$ and is of the form 
	\[V^X_{\ell}=\{(X'\mid X)\mid X'\in \PP_{\ell}\}.\]
	for a given $X=(X_{\ell+1}\mid \ldots\mid X_{k})\in \prod_{i=\ell+1}^{k}\Q_{i}$.
\end{lemma}
\begin{proof}
	Notice that $|V^X_{\ell}|=N_\ell$ and that $\{V_{\ell}^X\mid X\in \prod_{i=\ell+1}^{k}\Q_{i}\}$ is a partition of $V$. Furthermore since $E_{\ell}=\bigcup_{i=0}^{\ell}\A_i$ any edge $e\in E_{\ell}$ is either completely contained in or disjoint to $V^X_{\ell}$.
	
	It is left to show that $V^X_{\ell}$ is connected.
	We prove this via induction over $\ell$. For $\ell=0$ this is clear because $|V^X_{0}|=1$. 
	For $\ell>0$ let $Y=(Y_{\ell}\mid X), Z=(Z_{\ell}\mid X)\in \prod_{i=\ell}^{k}\Q_{i}$. By the induction hypothesis we know that the sets $V_{\ell-1}^Y, V_{\ell-1}^Z$ are connected. To prove that $V_\ell^X$ is connected it is sufficient to show that there is a path from a point in $V_{\ell-1}^Y$ to a point in  $V_{\ell-1}^Z$. We show this claim by induction over the number $m$ of coordinates in which $Y$ and $Z$ differ. For $m=0$ there is nothing to show. If $m>0$ pick $1\leq s\leq N_{\ell-1}$ such that $y_{\ell s}\neq z_{\ell s}$ and let $P=\phi_{\ell-1}^{-1}(s)\in \prod_{i=0}^{\ell-1} \Q_i$. Consider the point $(P\mid Y_\ell\mid X)$ which is contained in $V_{\ell-1}^Y$ . This point is also contained in the set    
	\[\{(P\mid y_{\ell 1},\ldots, y_{\ell s-1},\star, y_{\ell s+1}, \ldots , y_{\ell N_{\ell-1}}\mid  X)\mid \star\in[\Gamma \cdot N_{\ell-1}]\}\in E_{\ell}.\]
	Thus there is an edge in $G_{\ell}$ connecting a point in $V_{\ell-1}^Y$ to a point in $V_{\ell-1}^{Y'}$ with $Y'=(y_{\ell 1},\ldots, y_{\ell s-1}, z_{\ell s}, y_{\ell s+1},\ldots, y_{N_{\ell-1}}\mid X)$. Now $Y'$ and $Z$ differ in $m-1$ coordinates, thus there is a path between two points in $V_{\ell-1}^{Y'}$ and $V_{\ell-1}^Z$ by induction hypothesis. If we combine this with the induction hypothesis that $V_{\ell-1}^{Y'}$ is connected this yields the claim (see Figure~\ref{fig:conn_comp} for an illustration).  	
	\begin{figure}
		\centering
		\includegraphics[]{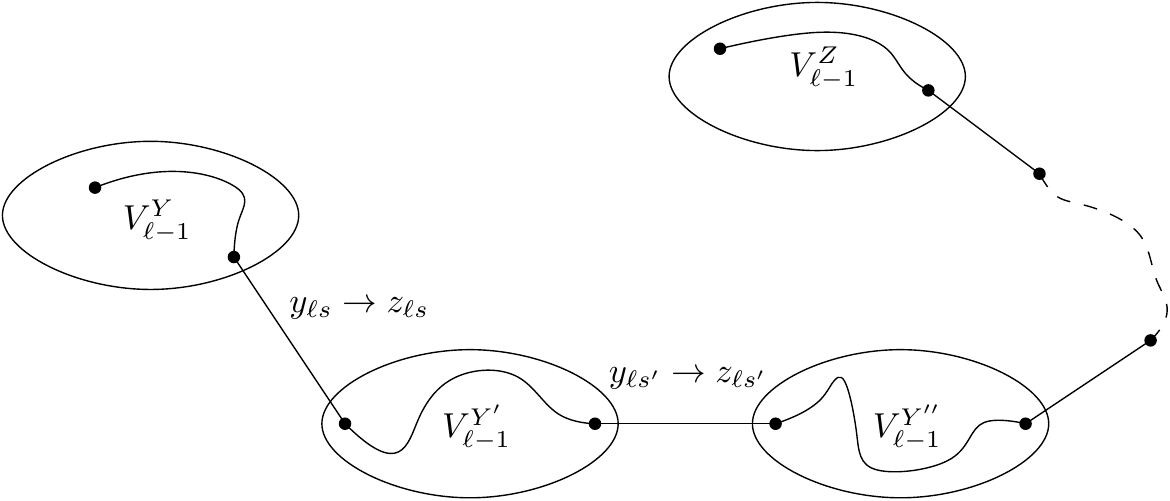}
		\caption{Here we see the construction of the path. It corresponds to changing the coordinates of $Y$ successively until they match $Z$. We use an edge in $\A_\ell$ to change $y_{ls}$ to $z_{ls}$, next we change $y_{ls'}$ to $z_{ls'}$ and proceed like this until we obtain $Z$. The respective edges are then connected to a path from $V_{\ell-1}^X$ to $V_{\ell-1}^Z$.}
		\label{fig:conn_comp}
	\end{figure}
\end{proof}

\begin{lemma}
	\label{lem:opt_cl}
	Any clustering of $(\X,\PP,d)$ with less than $\frac{N_k}{N_{\ell-1}}$ clusters costs at least $\ell$ if $\cost\in\{\diam, \drad\}$ and $\ell/2$ if $\cost=\rad$.
\end{lemma}
\begin{proof}
	The shortest path in $G$ between any two points which lie in different connected components of $G_{\ell-1}$ must contain an edge of weight $\geq \ell$. Thus any set of points $M\subset V$ which is disconnected in $G_{\ell-1}$ has diameter $\geq \ell$. Remember that the discrete radius of $M$ is given by $\drad(M)=\min_{c\in M}\max_{p\in M}d(p,c)$. For every possible choice of $c\in M$ there exists a point $p\in M$ which is not in the same connected component of $G_{\ell-1}$ as $c$, thus $d(c,p)\geq \ell$ and therefore $\drad(M)\geq \ell$ and $\rad(M)\ge \diam(M)/2 \ge \ell/2$.
	
	We conclude that if $\cost\in \{\diam,\drad\}$ any cluster of cost smaller than $\ell$ is contained in one of the sets $V_{\ell-1}^X$ for some $X\in \prod_{i=\ell}^k\Q_i$ by Lemma~\ref{lem2:construction_components} and any clustering with less than $\big|\prod_{i=\ell}^k\Q_i\big|$ clusters costs at least $\ell$. By the same argument if $\cost=\rad$ any cluster of cost smaller than $\ell/2$ is contained in one of the sets $V_{\ell-1}^X$ for some $X\in \prod_{i=\ell}^k\Q_i$ by Lemma~\ref{lem2:construction_components} and any clustering with less than $\big|\prod_{i=\ell}^k\Q_i\big|$ clusters costs at least $\ell/2$.
	Since 
	\begin{align*}
	\Big|\prod_{i=\ell}^k\Q_i\Big|=\frac{\big|\prod_{i=0}^k\Q_i\big|}{\big|\prod_{i=0}^{\ell-1}\Q_i\big|}=\frac{N_k}{N_{\ell-1}}
	\end{align*}
	this proves the lemma.
\end{proof}

\begin{corollary}
	\label{cor:opt_cl}
	For $1\leq \ell\leq k$ and $\cost\in\{\diam,\rad,\drad\}$ the clustering $\A_{\ell}$ is an optimal $\frac{N_k}{\Gamma N_{\ell-1}}$-clustering for the instance $(\X,\PP,d)$. Furthermore $\diam(\A_\ell)=\drad(\A_\ell)=\ell$ and $\rad(\A_\ell)=\ell/2.$
\end{corollary}
\begin{proof}
	If $\cost\in\{\diam,\drad\}$ we obtain by definition of $(\X,\PP,d)$ that $\cost(\A_\ell)\leq \ell$. If $\cost=\rad$  we obtain that $\cost(\A)\leq\ell/2$ by picking $v_A\in\X\backslash\PP$ as center for $A\in\A_\ell$. On the other hand  $|\A_{\ell}|=\frac{N_k}{\Gamma N_{\ell-1}}< \frac{N_k}{N_{\ell-1}}$ and thus $\cost(\A_{\ell})\geq \ell$ if $\cost\in\{\diam,\drad\}$ and $\cost(\A_{\ell})\geq \ell/2$ for $\cost=\rad$ by Lemma~\ref{lem:opt_cl}.
\end{proof}

\subsection{Characterization of Hierarchical Clusterings} 
Let from now on $\Hr=(\Ht_{N_k},\ldots, \Ht_{1})$ denote a hierarchical clustering of $(\X,\PP,d)$. 
We introduce the notion of \emph{bad clusters} in $\Ht_{\frac{N_k}{\Gamma N_{\ell-1}}}$ which are clusters whose cost increases repeatedly, as we will see later. In this section we prove the existence of such clusters in $\Hr$ and we give a lower bound on their cost.

\begin{definition}
	We call all clusters $C \in \Ht_{N_k}$ bad at time $0$ and denote by $\Ker_0(C)=C$ the kernel of $C$ at time $0$ and set $\Bad(0)=\Ht_{N_k}$.
	
	For $1\leq \ell\leq k$ we say that a cluster $C\in \Ht_{\frac{N_k}{\Gamma N_{\ell-1}}}$ is anchored at $\ell\leq \ell'\leq k$ if the set $\bigcup_{D\in \Bad(\ell-1)\colon D\subset C}\Ker_{\ell-1}(D)$ is
	\begin{enumerate}
		\item connected in $G_{\ell'}$,
		\item disconnected in $G_{\ell'-1}$.
	\end{enumerate}
	We call $C$ bad at time $\ell$ if $C$ is anchored at some $\ell'\geq\ell$.
	We denote by $\Bad(\ell)\subset \Ht_{\frac{N_k}{\Gamma N_{\ell-1}}}$ the set of all bad clusters at time $\ell$. If $C$ is bad we define the kernel of $C$ as the union of all kernels of bad clusters at time $\ell-1$ contained in $C$, i.e., 
	\[\Ker_\ell(C)=\bigcup_{D\in \Bad(\ell-1)\colon D\subset C}\Ker_{\ell-1}(D).\]
	All clusters in $\Ht_{\frac{N_k}{\Gamma N_{\ell-1}}}\backslash\Bad(\ell)$ are called good.
\end{definition}
\begin{figure}
	\centering
	\includegraphics[scale=0.65]{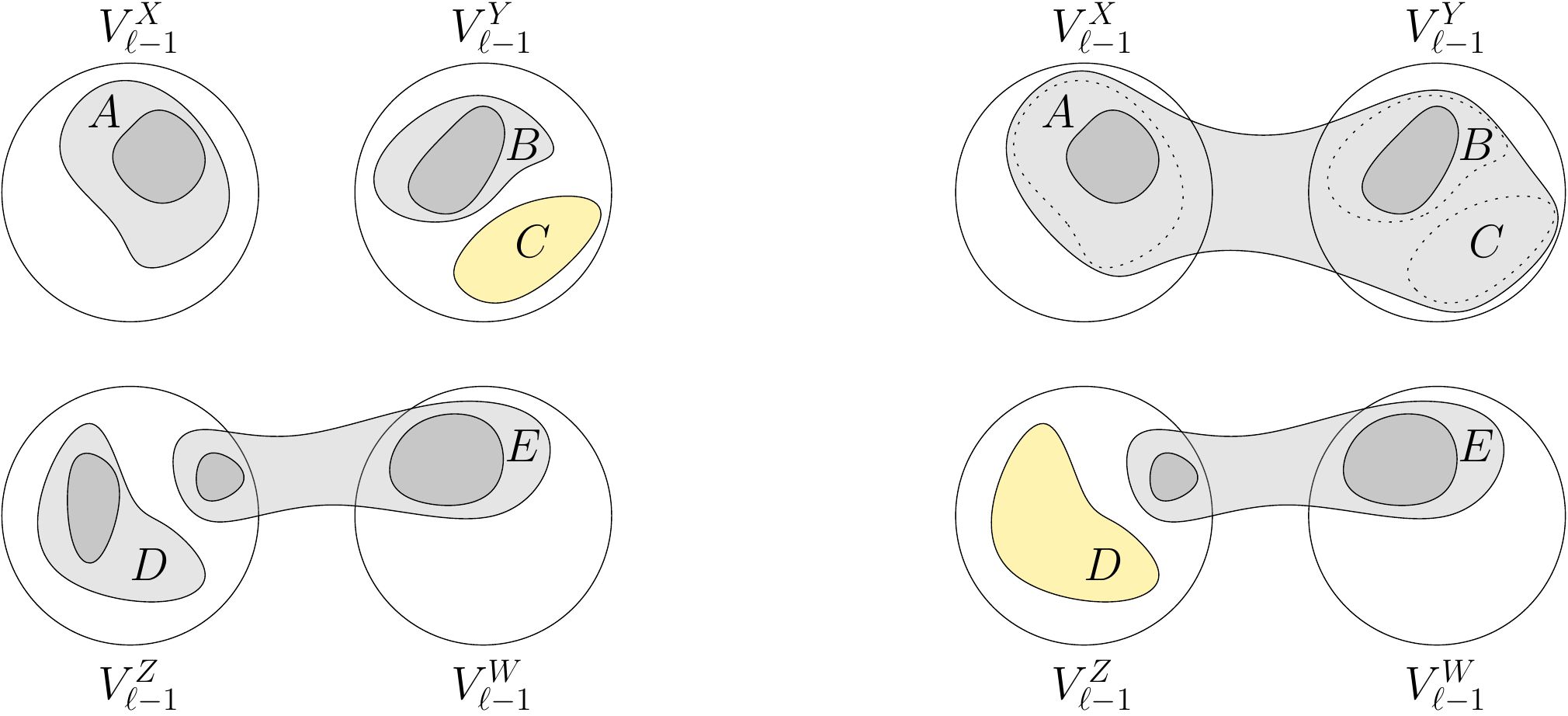}
	\caption{An illustration of the evolution of good and bad clusters: In the example, we see five clusters at time $\ell-1$. The clusters $A,B,D,E$ are assumed to be bad, with their kernels depicted in dark gray, while $C$ is assumed to be a good cluster. At time $\ell$, clusters $A,B$ and $C$ are merged. The resulting cluster is bad because the kernels of $A$ and $B$ lie in different connected components of $G_{\ell-1}$. Clusters $D$ and $E$ are still present at time $\ell$, but now $D$ is a good cluster since its kernel is completely contained in $V_{\ell-1}^Z$, while $E$ is still bad, since its kernel is disconnected in $G_{\ell-1}$.}
	\label{fig:kernel}
\end{figure}

\begin{lemma}
	\label{lem2:hierarchical_size_good}
	Let $C$ be a good cluster at time $1\leq \ell\leq k$ and \[W=\bigcup_{D\in\Bad(\ell-1)\colon D\subset C}\Ker_{\ell-1}(D),\] then $W$ is connected in $G_{\ell-1}$ and thus $|W|\leq N_{\ell-1}$. 
\end{lemma}
\begin{proof}
	Suppose $W$ is disconnected in $G_{\ell-1}$. Since $G_k=G$ is connected, there must be a time $\ell'\geq\ell$ such that $W$ is connected in $G_{\ell'}$ and disconnected in $G_{\ell'-1}$. But then $C$ is a bad cluster at time $\ell$ which is anchored at $\ell'$ in contradiction to our assumption. Thus $W$ is connected in $G_{\ell-1}$. By Lemma~\ref{lem2:construction_components} we know that every connected component in $G_{\ell-1}$ is of size $N_{\ell-1}$.
\end{proof}

The example in Figure~\ref{fig:kernel} shows that a bad cluster at time $\ell$ can contain clusters which are good at time $\ell-1$. However we are only interested in points that are contained exclusively in bad clusters at any time $t<\ell$. The set $\Ker_\ell(C)$ contains exactly such points.

We will use two crucial properties to prove the final lower bound on the approximation factor of any hierarchical clustering $\Hr$ of $(\X,\PP,d)$. We first observe that bad clusters exist in $\Hr$ for every time-step $1\leq \ell\leq k$ and second that these clusters have a large cost compared to the optimal clustering. 
\begin{lemma}
	\label{lem2:hierarchical_number_bad}
	For all $0\leq \ell\leq k$ we have 
	\[\sum_{C\in\Bad(\ell)}|\Ker_\ell(C)|\geq \frac{\Gamma-\ell}{\Gamma}N_k.\]
\end{lemma}
\begin{proof}
	We prove this via induction over $\ell$.  
	For $\ell=0$ this is clear since \[\bigcup_{C\in\Bad(0)}\Ker_0(C)=\PP_k.\]
	
	Now suppose that  $\ell>0$ and that
	\[\sum_{C\in\Bad(\ell)}|\Ker_\ell(C)|< \frac{\Gamma-\ell}{\Gamma}N_{k}.\] By induction hypothesis we know that 
	\[\sum_{C\in\Bad(\ell-1)}|\Ker_{\ell-1}(C)|\geq \frac{\Gamma-\ell+1}{\Gamma}N_k.\]
	Thus the number of points which are in the kernel of a bad cluster at time $\ell-1$ but not at time $\ell$ is larger than
	\[\frac{\Gamma-\ell+1}{\Gamma}N_k-\frac{\Gamma-\ell}{\Gamma}N_k=\frac{N_k}{\Gamma}.\]
	In other words these are points that are in the kernel of a bad cluster at time $\ell-1$ but contained in a good cluster at time $\ell$. Now we use that any good cluster at time $\ell$ can contain only $N_{\ell-1}$ such points by Lemma~\ref{lem2:hierarchical_size_good}.
	Thus the number of good clusters is greater than  
	\[\frac{N_k}{\Gamma}\cdot \frac{1}{N_{\ell-1}}= \frac{N_k}{\Gamma N_{\ell-1}}.\]
	We obtain that $\Ht_{\frac{N_k}{\Gamma N_{\ell-1}}}$ contains more than $\frac{N_k}{\Gamma N_{\ell-1}}$ clusters, which is not possible. 
\end{proof}

An immediate consequence of Lemma~\ref{lem2:hierarchical_number_bad} is the existence of bad clusters at time $\ell$ for any $0\leq \ell\leq k$. To prove that their (discrete) radius and diameter is indeed large we need a lower bound on the distance between two points $X,Y\in \PP$ that lie in different connected components of $G_{j-1}$ for some $1\leq j \leq k$. 

Suppose that the points $X$ and $Y$ only differ in one coordinate, i.e., there is a $1\leq s\leq N_{j-1}$ such that $x_{js}\neq y_{js}$, while $X$ and $Y$ agree in all other coordinates. 
There is only one edge in $G_j$ connecting $V_{j-1}^{X_{[j:k]}}$ with $V_{j-1}^{Y_{[j:k]}}$. Let $P=\phi_{j-1}^{-1}(s)$, then this edge connects the points $(P\mid X_{[j:k]})$ and $(P\mid Y_{[j:k]})$. If we connect $X$ to $(P\mid X_{[j:k]})$ and $(P\mid Y_{[j:k]})$ to $Y$ via a shortest path, this results in a path from $X$ to $Y$, see Figure~\ref{fig:distance}. We show that this path is indeed a shortest path between $X$ and $Y$ and generalize this to arbitrary $X$ and $Y$ which are disconnected in $G_{j-1}$.

\begin{figure}
	\centering
	\includegraphics[scale=1]{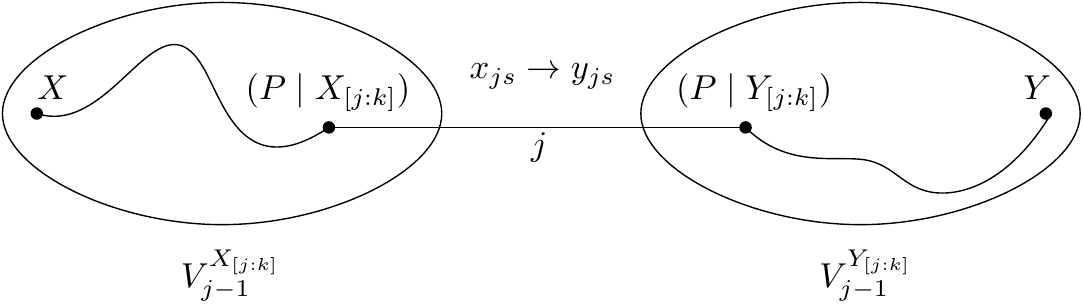}
	\caption{A shortest path between $X$ and $Y$. It consists of two shortest paths inside the connected components of $G_{j-1}$ and the unique edge of weight $j$ between these components. }
	\label{fig:distance}
\end{figure}

\begin{restatable}{lemma}{LemmaDistance}
	\label{lem2:distance_comp}
	Let $X,Y\in \PP$ be two points and suppose there is $1\leq j\leq k$ and $1\leq s\leq N_{j-1}$ such that $x_{js}\neq y_{js}$. Let $P=\phi_{j-1}^{-1}(s)\in \prod_{i=0}^{j-1}\Q_i$.
	Then \[d(X,Y)\geq d\big (X, (P\mid X_{[j:k]})\big)+j+ d\big(Y, (P\mid Y_{[j:k]})\big).\] 
\end{restatable}
\begin{proof}
Observe that if two points in $G$ are connected by an edge they differ in exactly one coordinate. 
Since $x_{js}\neq y_{js}$ any shortest path connecting $X$ and $Y$ must contain two consecutive points $Z,Z'$ with $Z=(P\mid Z_j\mid\ldots\mid Z_k)$ and $Z'=(P\mid Z'_j\mid\ldots\mid Z'_k)$ such that $z_{js}=x_{js}, z'_{js}=y_{js}$ and $Z$ agrees with $Z'$ in all remaining coordinates. We obtain
\[d(X,Y)= d\big (X, Z\big)+d(Z,Z')+ d\big(Z', Y\big)=d\big (X, Z\big)+j+ d\big(Z', Y\big).\]
It is now left to show that $d(X,Z)\geq d\big(X, (P\mid X_{[j:k]})\big)$ and $d(Y,Z')\geq d\big(Y, (P\mid Y_{[j:k]})\big)$. 
To prove this we consider a shortest path $V^1,\ldots, V^t$ connecting $V^1=X$ with $V^t=Z$. Let $W^i=(V^i_{[0:j-1]}\mid X_{[j:k]})$ for $i=1,\ldots, t$. We claim that $W^i$ is connected to $W^{i+1}$ by an edge in $G$ and that $d(V^i,V^{i+1})\geq d(W^i,W^{i+1})$ for all $1\leq i\leq t-1$. So let $1\leq i\leq t-1$, we know that $V^i$ and $V^{i+1}$ differ in exactly one coordinate. If they differ at a coordinate in row $r\geq j$ we have $W^i=W^{i+1}$ and thus the claim holds. Otherwise let $u=\phi_{r-1}(V^i_{[0:r-1]})$ then $V^i$ and $V^{i+1}$ satisfy $v^{i}_{ru}\neq v^{i+1}_{ru}$ and $d(V^i,V^{i+1})=r$. Since $r\leq j-1$ we obtain that $W^{i}$ is connected to $W^{i+1}$ by the edge 
\[\{(V^{i}_{[0:r-1]}\mid v^i_{r1},\ldots,v^i_{ru-1},\star, v^i_{ru+1},\ldots, v^i_{rN_{r-1}}\mid W^i_{[r+1:k]})\mid \star\in[\Gamma N_{r-1}] \},\]
which has weight $r$. This yields the claim. 

Observe that $W^1=X$ and $W^t=(P\mid X_{[j:k]})$ and that
\begin{align*}
d\big(X, (P\mid X_{[j:k]})\big) \leq\sum_{i=1}^{t-1}d(W^i, W^{i+1})\leq \sum_{i=1}^{t-1}d(V^i, V^{i+1})=d(X,Z).
\end{align*} 
Analogously one can show $d(Y,Z')\geq d\big(Y, (P\mid Y_{[j:k]})\big)$ and obtains 
\[d(X,Y)=d\big (X, Z\big)+j+ d\big(Z', Y\big)\geq d\big(X, (P\mid X_{[j:k]})\big)+j+d\big(Y, (P\mid Y_{[j:k]})\big).\]
\end{proof}

We now define the so called \emph{anchor set} $\Anc_\ell(C)$ of a bad cluster $C$ at time $\ell$. If $C$ is anchored at $\ell'$ then $\Anc_\ell(C)$ is the union of $\ell'$ and the anchor set of some bad cluster $D\subset C$ at time $\ell-1$. If we choose $D$ appropriately the sum of anchors in $\Anc_\ell(C)$ is a lower bound on the discrete radius of $C$, as we show later. It is clear that $\ell'$ itself is a lower bound on the discrete radius since $\Ker_\ell(C)$ is disconnected in $G_{\ell'-1}$ by definition.
If we additionally assume that the discrete radius of $D$ is large, e.g., lower bounded by the sum of anchors in $\Anc_{\ell-1}(D)$, then it is reasonable to assume that the discrete radius of $C$ is lower bounded by some function in $\ell'$ and the sum of anchors in $\Anc_{\ell-1}(D)$. Before proving this we give a formal definition of $\Anc_\ell(C)$ and how to choose $D$.

\begin{definition}
	\label{def2:anc_set}
	Let $1\leq \ell\leq k$ and $C$ be a bad cluster at time $\ell$ which is anchored at $\ell'\geq \ell$. If $\ell=1$ we define the anchor set of $C$ as $\Anc_1(C)=\{\ell'\}$ and set $\prev(C)=\{X\}$ for some $X\in C$.
	
	For $\ell>1$ we distinguish two cases.
	\begin{description}
		\item[Case 1:] $C$ contains a bad cluster $D$ which is bad at time $\ell-1$ and anchored at $\ell'$. We then set  $\Anc_{\ell}(C)=\Anc_{\ell-1}(D)$ and $\prev(C)=D$.	
		\item[Case 2:] $C$ does not contain such a cluster. Then let $D\subset C$ be a bad cluster at time $\ell-1$ minimizing \[\sum_{a\in\Anc_{\ell-1}(D)}a\]
		among all clusters $D'\in \Bad(\ell-1)$ with $D'\subset C$. 
		We set $\Anc_\ell(C)=\Anc_{\ell-1}(D)\cup\{\ell'\}$ and $ \prev(C)=D$.
	\end{description}
\end{definition}

Observe that in Case~2 of the previous definition, the bad cluster~$D$ must be anchored at some~$\ell_D < \ell'$.
\begin{lemma}
	\label{lem2:kernel_subset}
	Let $1\leq \ell\leq k$ and $C$ be a bad cluster at time $\ell$. If $C$ contains a cluster $D$ which is bad at time $\ell-1$ then $\Ker_{\ell-1}(D)\subset \Ker_\ell(C)$.
\end{lemma}
\begin{proof}
	Since $D\in\Bad(\ell-1)$ and $D\subset C$, we get 
	\[\Ker_{\ell-1}(D)\subset \bigcup_{D'\subset \Bad(\ell-1)\colon D'\subset C}\Ker_{\ell-1}(D')=\Ker_{\ell}(C).\]
\end{proof}
With the help of Lemma~\ref{lem2:distance_comp} we are able to show how the discrete radius and diameter of a bad cluster, depends on the sum of anchors.
\begin{restatable}{lemma}{LemmaAncRadius}
	\label{lem2:radius_bad}
	Let $1\leq \ell\leq k$ and $C$ be a bad cluster at time $\ell$ anchored at $\ell'$. Then for any point $Z\in \PP$  there is $X\in \Ker_\ell(C)$ such that
	\begin{align*}
	d(Z,X)\geq \sum_{a\in\Anc_{\ell}(C)}a.
	\end{align*}
\end{restatable}
\begin{proof}
	Let $Z\in \PP$ and suppose that $C$ is a bad cluster at time $\ell$ anchored at $\ell'$.
	We prove the lemma via induction over $\ell$. For $\ell=1$ we know that $\Ker_\ell(C)$ is disconnected in $G_{\ell'-1}$ by definition. Thus there is a point $X\in\Ker_\ell(C)$ which is disconnected from $Z$ in $G_{\ell'-1}$ yielding 
	\[d(Z,X)\geq \ell'=\sum_{a\in\Anc_{1}(C)}a.\] 
	
	Let $\ell>1$. If $D=\prev(C)$ is anchored at $\ell'$ we apply Lemma~\ref{lem2:kernel_subset} to observe that $\Ker_{\ell-1}(D)\subset \Ker_\ell(C)$. By induction hypothesis the lemma holds for $D$. Since $\Anc_\ell(C)=\Anc_{\ell-1}(D)$ the lemma also holds for $C$.
	
	Otherwise let $D=\prev(C)$ be anchored at $\ell_D<\ell'$. We know that $\Ker_\ell(C)$ is disconnected in $G_{\ell'-1}$. On the other hand $\Ker_{\ell-1}(D)$ is connected in $G_{\ell'-1}$ since $\ell_D<\ell'$. Thus there is $V\in \Ker_{\ell}(C)$ which is disconnected from $\Ker_{\ell-1}(D)$ in $G_{\ell'-1}$. Let $E\subset C$ be the cluster at time $\ell-1$ which contains $V$. Since $V\in \Ker_\ell(C)$ we know that $E$ is a bad cluster at time $\ell-1$ anchored at $\ell_E<\ell'$. We know that $\Ker_{\ell-1}(E)$ is connected in $G_{\ell'-1}$ and lies in a different connected component than $\Ker_{\ell-1}(D)$. Thus $Z$ is disconnected from $\Ker_{\ell-1}(D)$ or $\Ker_{\ell-1}(E)$ in $G_{\ell'-1}$.
	
	We assume without loss of generality that $Z$ is disconnected from $E$ in $G_{\ell'-1}$. Since $\Ker_{\ell-1}(E)$ is connected in $G_{\ell'-1}$ we know by Lemma~\ref{lem2:construction_components} that $(P\mid Y_{[\ell':k]})=(P\mid Y'_{[\ell':k]})$ for all $Y,Y'\in \Ker_{\ell-1}(E)$.
	Also by Lemma~\ref{lem2:construction_components} there is $\ell'\leq r\leq k$ and $1\leq s\leq N_{r-1}$ such that $z_{rs}\neq y_{rs}$ for all $Y\in \Ker_{\ell-1}(E)$. Let $P=\phi_{r-1}^{-1}(s)$. 
	Thus we know by induction hypothesis that there is a point $X\in\Ker_{\ell-1}(E)\subset \Ker_{\ell}(C)$ with
	\[d(X, (P\mid X_{[r:k]}))\geq \sum_{a\in \Anc_{\ell-1}(E)}a.\] Figure~\ref{fig:rad_bound} shows an exemplary path between $X$ and $Z$.
	
	We apply Lemma~\ref{lem2:distance_comp} to see that
	\begin{align*}
	d(Z,X)
	&\geq d\big (Z, (P\mid Z_{[r:k]})\big)+r+ d\big(X, (P\mid X_{[r:k]})\big)\\
	&\geq r+\sum_{a\in \Anc_{\ell-1}(E)}a\\
	&\geq \ell'+\sum_{a\in \Anc_{\ell-1}(E)}a\\
	&\geq \sum_{a\in \Anc_{\ell}(C)}a
	\end{align*}
	Here the last inequality follows from the minimality of $\sum_{a\in\Anc_{\ell-1}(D)}a$ among all clusters $D'\in \Bad(\ell-1)$ with $D'\subset C$. 
	
	If $Z$ is disconnected from $D$ in $G_{\ell'-1}$ our argument still works after replacing $E$ by $D$.
\end{proof}
\begin{figure}
	\centering
	\includegraphics[scale=1]{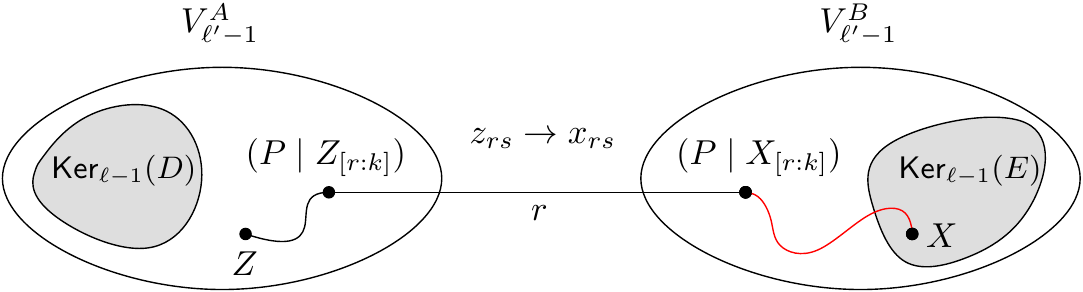}
	\caption{Shows the special case where $Z_{[r:k]}$ and $Y_{[r:k]}$ only differ in the $rs$-coordinate. The length of the red path is lower bounded by $\sum_{a\in \Anc_{\ell-1}(E)}a$.}
	\label{fig:rad_bound}
\end{figure}

\begin{figure}
	\centering
	\includegraphics[scale=1]{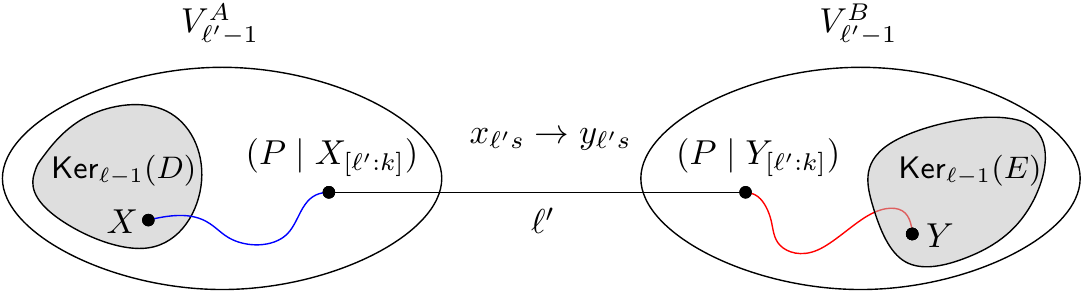}
	\caption{Shows the special case where $X_{[\ell':k]}$ and $Y_{[\ell':k]}$ only differ in the $\ell's$-coordinate. The length of the blue path is lower bounded by $\sum_{a\in \Anc_{\ell-1}(D)}a$, while the length of the red path is lower bounded by $\sum_{a\in\Anc_{\ell-1}(E)}a$.}
	\label{fig:diam_bound}
\end{figure}

\begin{restatable}{lemma}{LemmaAncDiameter}
	\label{lem2:diameter_bad}
	Let $1\leq \ell\leq k$ and $C$ be a bad cluster at time $\ell$ anchored at $\ell'$. Then there are two points $X,Y\in \Ker_\ell(C)$ such that
	\begin{align*} 
	d(X,Y)\geq \ell'+2\sum_{a\in\Anc_{\ell}(C)\backslash\{\ell'\}}a.
	\end{align*}
\end{restatable}
\begin{proof}
	Suppose that $C$ is a bad cluster at time $\ell$ anchored at $\ell'$.
	We prove the lemma via induction over $\ell$. For $\ell=1$ we know that $\Ker_\ell(C)$ is disconnected in $G_{\ell'-1}$ by definition. Thus there are two points $X,Y\in\Ker_\ell(C)$ that are disconnected in $G_{\ell'-1}$ yielding 
	\[d(X,Y)\geq \ell'=\ell'+2\sum_{a\in\Anc_{1}(C)\backslash\{\ell'\}}a.\] 
	
	Let $\ell>1$. If $D=\prev(C)$ is anchored at $\ell'$ we apply Lemma~\ref{lem2:kernel_subset} to observe that $\Ker_{\ell-1}(D)\subset \Ker_\ell(C)$. By induction hypothesis the lemma holds for $D$. Since $\Anc_\ell(C)=\Anc_{\ell-1}(D)$ the lemma also holds for $C$.
	
	Otherwise let $D=\prev(C)$ be anchored at $\ell_D<\ell'$. We know that $\Ker_\ell(C)$ is disconnected in $G_{\ell'-1}$ and $\Ker_{\ell-1}(D)$ is connected in $G_{\ell'-1}$. Thus there is $V\in \Ker_{\ell}(C)$ which is disconnected from $\Ker_{\ell-1}(D)$ in $G_{\ell'-1}$. Let $E\subset C$ be the cluster at time $\ell-1$ which contains $V$. We know that $E$ is a bad cluster at time $\ell-1$ anchored at $\ell_E<\ell'$. Furthermore $\Ker_{\ell-1}(E)$ is connected in $G_{\ell'-1}$ and lies in a different connected component than $\Ker_{\ell-1}(D)$.

	Since $\Ker_{\ell-1}(D)$ and $\Ker_{\ell-1}(E)$ are disconnected in $G_{\ell'-1}$ but connected in $G_{\ell'}$, there must be $1\leq s\leq N_{\ell'-1}$ such that for all $U\in\Ker_{\ell-1}(D)$ and $T\in\Ker_{\ell-1}(E)$ we have $u_{\ell's}\neq t_{\ell' s}$ by Lemma~\ref{lem2:construction_components}. Let $P=\phi^{-1}_{\ell'-1}(s)$, we know by Lemma~\ref{lem2:distance_comp} that 
	\begin{align*}
	d(U,T)\geq d(U,(P\mid U_{[\ell':k]}))+\ell'+d(T,(P\mid T_{[\ell':k]})).
	\end{align*}
	Let $U\in\Ker_{\ell-1}(D)$ and $T\in\Ker_{\ell-1}(E)$. We know by Lemma~\ref{lem2:radius_bad} that for any two points $Z=(P\mid U_{[\ell':k]})$ and  $Z'=(P\mid T_{[\ell':k]})$ there must be $X\in\Ker_{\ell-1}(D)$ and $Y\in\Ker_{\ell-1}(E)$ such that 
	\[d(X,Z)\geq \sum_{a\in\Anc_{\ell-1}(D)}a\]
	and 
	\[d(Y,Z')\geq \sum_{a\in\Anc_{\ell-1}(E)}a.\]
	We use Lemma~\ref{lem2:construction_components} to observe that  $Z=(P\mid X_{[\ell':k]})$ and $Z'=(P\mid Y_{[\ell':k]})$ because $X$ is connected to $U$ and $Y$ is connected to $T$ in $G_{\ell'-1}$. Figure~\ref{fig:diam_bound} shows an exemplary path between $X$ and $Y$.
	Thus
	\begin{align*}
	d(X,Y)
	&\geq d(X,(P\mid X_{[\ell':k]}))+\ell'+d(Y,(P\mid Y_{[\ell':k]}))\\
	&\geq d(X,Z)+\ell'+d(Y,Z')\\
	&\geq \ell'+\sum_{a\in\Anc_{\ell-1}(D)}a+\sum_{a\in\Anc_{\ell-1}(E)}a\\
	&\geq \ell'+2\sum_{a\in \Anc_{\ell}(C)\backslash\{\ell'\}}a
	\end{align*} 
	Here the last inequality follows from the minimality of $\sum_{a\in\Anc_{\ell-1}(D)}a$ among all clusters $D'\in \Bad(\ell-1)$ with $D'\subset C$. 
\end{proof}
\subsection{Comparison to Optimal Clusterings}
Our initial motivation was to construct an instance where any hierarchical clustering has a high approximation ratio. If we consider an arbitrary time $1\leq \ell\leq k$ then the hierarchical clustering $\Hr$ on $(\X,\PP,d)$ may be even optimal at time $\ell$.
Thus the bounds which we develop in Lemma~\ref{lem2:radius_bad} and Lemma~\ref{lem2:diameter_bad} on the discrete radius and diameter of bad clusters are useless without linking the cost of a bad cluster at time $\ell$ to the cost of bad clusters at other time steps. Therefore we construct a sequence of clusters $C_1\subset C_2\ldots\subset C_k$ where $C_i$ is a bad cluster at time $i$ such that $\Anc_1(C_1)\subset\Anc_{2}(C_2)\subset\ldots\subset\Anc_k(C_k)$. We then show with the help of Lemma~\ref{lem2:radius_bad} and Lemma~\ref{lem2:diameter_bad} that at least one of these clusters has a high discrete radius and diameter compared to the optimal cost.
\begin{lemma}
	\label{lem2:inc_cluster_seq}
	Let $C_k$ be a bad cluster at time $k$. For $1\leq i\leq k-1$ we define $C_i=\prev(C_{i+1})$. For all $1\leq i\leq k-1$ cluster $C_i$ is a bad at time $i$ and one of the following two cases occurs:
	\begin{enumerate}
		\item $\Anc_i(C_i)= \Anc_{i+1}(C_{i+1})$,
		\item $\Anc_{i+1}(C_{i+1})\backslash \{\ell\}=  \Anc_i(C_i)$, where $\ell= \max \Anc_{i+1}(C_{i+1})$.
	\end{enumerate} 
\end{lemma}
\begin{proof}
	For $i=k$ cluster $C_k$ is bad at time $k$ by assumption. If $C_{i+1}$ is a bad cluster at time $i+1$ then $C_i=\prev(C_{i+1})$ is a bad cluster at time $i$, by definition of $\prev$.
	
	Let $C_i$ be anchored at $\ell'\geq i$ and $C_{i+1}$ be anchored at $\ell\geq i+1$. Since $\Ker_{i}(C_i)\subset \Ker_{i+1}(C_{i+1})$ by Lemma~\ref{lem2:kernel_subset}, we know that $\ell'\leq \ell$.
	If $\ell'=\ell$ we obtain by Definition~\ref{def2:anc_set}, that $\Anc_{i}(C_i)=\Anc_{i+1}(C_{i+1})$, so the lemma holds in this case.
	
	If $\ell'<\ell$ we know by Definition~\ref{def2:anc_set} that $\Anc_{i}(C_i)=\Anc_{i+1}(C_{i+1})\backslash\{\ell\}$. So the lemma also holds in this case. 
\end{proof}
\begin{corollary}
	\label{cor:anc_sequence}
	Let $C_k$ be a bad cluster at time $k$. For $1\leq i\leq k-1$ we define $C_i=\prev(C_{i+1})$. Let $\Anc_k(C_k)=\{\ell_1,\ldots,\ell_s\}$ such that $\ell_{t-1}<\ell_t$ for all $2\leq t\leq s$ and let $\ell_0=0$. Then for any $1\leq t\leq s$ and for any $i$ with $\ell_{t-1}< i \leq \ell_t$, we have $\{\ell_1,\ldots,\ell_t\}\subset \Anc_{i}(C_i)$.
\end{corollary}
\begin{proof}
	We prove this via induction over $i$, starting from $i=k$ in decreasing order.
	There is nothing to show for $i=k$. For $i<k$ we distinguish two cases. If $\Anc(C_i)=\Anc_{i+1}(C_{i+1})$, the lemma follows from the induction hypothesis. 
	
	Otherwise remember that $\Anc_i(C_i)\subset \Anc_k(C_k)$ and $\ell_{t-1}<i$. Thus we know that $\max \Anc_{i}(C_{i})\in \{\ell_t,\ldots,\ell_s\}$ and therefore $\ell_t\leq \max\Anc_i(C_i)$.
	By Lemma~\ref{lem2:inc_cluster_seq} we know that $\Anc_i(C_i)=\Anc_{i+1}(C_{i+1})\backslash\{\ell\}$, where $\ell=\max \Anc_{i+1}(C_{i+1})$. Thus
	\[\ell_t\leq  \max \Anc_{i}(C_{i})<\max\Anc_{i+1}(C_{i+1})=\ell.\]
	By induction hypothesis we obtain
	\[\{\ell_1,\ldots,\ell_t\} \subset \Anc_{i+1}(C_{i+1})\backslash\{\ell\}=\Anc_i(C_i).\qedhere\]
\end{proof}
In consequence of Corollary~\ref{cor:anc_sequence} we obtain for $\ell_{t-1}< i \leq \ell_t$ in combination with Lemma~\ref{lem2:radius_bad} and Lemma~\ref{lem2:diameter_bad} the following lower bound on the approximation guarantee of $\Hr$ at time $i$
\[
\frac{\rad\Big(\Ht_{\frac{N_k}{\Gamma N_{i-1}}}\Big)}{\rad(\A_i)} =
\frac{2\rad\Big(\Ht_{\frac{N_k}{\Gamma N_{i-1}}}\Big)}{2\rad(\A_i)} \geq
\frac{\diam\Big(\Ht_{\frac{N_k}{\Gamma N_{i-1}}}\Big)}{\diam(\A_i)}\geq \frac{\diam(C_i)}{\diam(\A_i)}\geq \frac{\ell_t+2\sum_{j=0}^{t-1}\ell_j}{i}\]
and 
\[\frac{\drad\Big(\Ht_{\frac{N_k}{\Gamma N_{i-1}}}\Big)}{\drad(\A_i)}\geq \frac{\drad(C_i)}{\drad(\A_i)}\geq \frac{\ell_t+\sum_{j=0}^{t-1}\ell_j}{i}\]
while the right term attains its maximum for $i=\ell_{t-1}+1$. Thus the approximation factor of $\Hr$ is lower bounded by 
$\max_{1\leq t\leq s}\frac{\ell_t+2\sum_{j=0}^{t-1}\ell_j}{\ell_{t-1}+1}$
for the diameter and 
$\max_{1\leq t\leq s}\frac{\ell_t+\sum_{j=0}^{t-1}\ell_j}{\ell_{t-1}+1}$ for the discrete radius. The next lemma shows that this is indeed our desired lower bound. 

\begin{restatable}{lemma}{lemmaNmbSeq}
	\label{lem2:number_sequence} For every $\epsilon>0$ there exists $k\in\mathbb N$ such that for every $s\in\mathbb N$ any sequence of $s+1$ numbers $(\ell_0,\ldots,\ell_s)\in \mathbb R_{\geq0}^{s+1}$ with $\ell_0=0$ and $\ell_s=k$ satisfies the following.  
	 \begin{enumerate}
	 	\item There exists $1\leq t\leq s$ such that for $\alpha_1=4-\epsilon$ and $\Delta_1=1$ we have 
	 	\begin{align*}
	 	\frac{\ell_t+\Delta_1\sum_{i=0}^{t-1}\ell_i}{\ell_{t-1}+1}>\alpha_1.
	 	\end{align*}
	 	\item There exists $1\leq t\leq s$ such that for $\alpha_2=3+2\sqrt{2}-\epsilon$ and $\Delta_2=2$ we have 
	 	\begin{align*}
	 	\frac{\ell_t+\Delta_2\sum_{i=0}^{t-1}\ell_i}{\ell_{t-1}+1}>\alpha_2.
	 	\end{align*}
	 \end{enumerate}
\end{restatable}

The proof of Lemma~\ref{lem2:number_sequence} can be found in Appendix~\ref{ap:lower_bound}.

\ThmLowerBound*
\begin{proof}
	Let $\epsilon>0$ and $k$ be the respective number from Lemma~\ref{lem2:number_sequence}. We claim that the approximation factor of any hierarchical clustering $\Hr=(\Ht_{N_k},\ldots,\Ht_{1})$ on the instance $(\X,\PP,d)$ is larger than $3+2\sqrt{2}-\epsilon$ if $\cost\in\{\diam,\rad\}$ and larger than $4-\epsilon$ if $\cost=\drad$. First we use Lemma~\ref{lem2:hierarchical_number_bad} to observe that there is a cluster $C_k\in\Ht_{\frac{N_k}{\Gamma N_{k-1}}}$ that is bad at time $k$. For $1\leq i\leq k-1$ we define $C_i=\prev(C_{i+1})$.
	Let $\Anc_k(C_k)=\{\ell_1,\ldots,\ell_s\}$ with $\ell_{t-1}<\ell_t$ for $2\leq t\leq s$ and let $\ell_0=0$. We know by Corollary~\ref{cor:anc_sequence}, that for any  $1\leq t\leq s$ and for $i=\ell_{t-1}+1$ we have $\{\ell_1,\ldots,\ell_t\}\subset\Anc_{i}(C_i)$. Let $\ell'=\max\Anc_i(C_i)$, we obtain by Lemma~\ref{lem2:diameter_bad} and Lemma~\ref{lem2:radius_bad} that
	\begin{align*}
	\diam(C_i)&\geq \ell'+2\sum_{a\in\Anc_i(C_i)\backslash\{\ell'\}}a \geq \ell_t+2\sum_{u=1}^{t-1}\ell_u,\\
	\rad(C_i)&\geq \frac{\diam(C_i)}{2}\geq \frac{\ell_t+2\sum_{u=1}^{t-1}\ell_u}{2},\\
	\drad(C_i)&\geq \sum_{a\in\Anc_i(C_i)}a\geq \sum_{u=1}^{t}\ell_u.
	\end{align*}
	Remember that by Corollary~\ref{cor:opt_cl} $\A_i$ is an optimal $\frac{N_k}{\Gamma N_{i-1}}$-clustering with $\cost(\A_i)=i$ if $\cost\in\{\diam,\drad\}$ and $\cost(\A_i)=i/2$ if $\cost=\rad$. We obtain
	\begin{align*}
	\frac{\rad(C_i)}{\rad(\A_i)}&=\frac{2\rad(C_i)}{2\rad(\A_i)}\geq \frac{\diam(C_i)}{\diam(\A_i)}\geq \frac{\ell_t+2\sum_{u=1}^{t-1}\ell_u}{\ell_{t-1}+1}\\
	\frac{\drad(C_i)}{\drad(\A_i)}&\geq \frac{\sum_{u=1}^{t}\ell_u}{\ell_{t-1}+1}
	\end{align*}
	which are lower bounds on the approximation factor of $\Hr$.
	
	We apply Lemma~\ref{lem2:number_sequence} on $(\ell_0,\ldots,\ell_s)$ to observe that there is $1\leq t'\leq s$ such that 
	\[\frac{\ell_{t'}+2\sum_{u=1}^{t'-1}\ell_u}{\ell_{t'-1}+1}>3+2\sqrt{2}-\epsilon\] and an $1\leq t''\leq s$ such that 
	\[\frac{\sum_{u=1}^{t''}\ell_u}{\ell_{t''-1}+1}> 4-\epsilon.\]
	This proves the theorem. 
\end{proof}

\section{Conclusions and Open Problems}

We have proved tight bounds for the price of hierarchy with respect to the diameter and (discrete) radius. It would be interesting to also obtain a better understanding of the price of hierarchy for other important objective functions like $k$-median and $k$-means. The best known upper bound is 16 for $k$-median~\cite{Dai14} and 32 for $k$-means~\cite{G20} but no non-trivial lower bounds are known. Closing this gap also for these objectives is a challenging problem for further research.

Another natural question is which approximation factors can be achieved by algorithms running in polynomial time. The algorithm we used in this article to prove the upper bounds is not a polynomial-time algorithm because it assumes that for each level~$k$ an optimal $k$-clustering is given. The approximation factors worsen if only approximately optimal clusterings are used instead. It is known that 8-approximate hierarchical clusterings can be computed efficiently with respect to the diameter and (discrete) radius~\cite{DasguptaL05}. It is not clear whether or not it is NP-hard to obtain better hierarchical clusterings. The only NP-hardness results come from the problems with given~$k$. Since computing a $(2-\epsilon)$-approximation for $k$-clustering with respect to the diameter and (discrete) radius is NP-hard, this is also true for the hierarchical versions. However, this is obsolete due to our lower bound, which shows that in general there does not even exist a $(2-\epsilon)$-approximate hierarchical clustering. 

\bibliography{literature}
\bibliographystyle{plainurl}
\appendix

\section{On the Lower Bound}
\label{ap:lower_bound}
\lemmaNmbSeq*
\begin{proof}
	Let $k,s\in \mathbb N$ and $j\in\{1,2\}$.
	We call a sequence $(a_0,\ldots,a_s)\in \mathbb R_{\geq0}^{s+1}$ \emph{feasible} if $a_0=0, a_s=k$ and for all $1\leq t\leq s$ we have 
	\begin{align}
		\label{eq:sum}
		\frac{a_t+\Delta_j\sum_{i=0}^{t-1}a_i}{a_{t-1}+1}\leq\alpha_j.
	\end{align}
	
	Our proof is divided in two parts. In the first part we argue that for all $k,s\in\mathbb N$ the existence of a feasible sequence $(\ell_0,\ldots,\ell_s)$ yields the existence of a feasible sequence $(b_0,\ldots,b_s)$ which satisfies (\ref{eq:sum}) for all $u+1\leq t\leq s$ with equality, where $u$ is the smallest number such that $b_u\neq 0$. In the second part we observe that there exists $k\in\mathbb N$ such that for all $s\in\mathbb N$ there is no feasible sequence $(a_0,\ldots,a_s)\in \mathbb R_{\geq0}^{s+1}$ which satisfies (\ref{eq:sum}) for all $u+1\leq t\leq s$ with equality, where $u$ is the smallest number such that $a_u\neq 0$. In combination both parts yield the lemma.
	
	\textbf{Part 1:}  Let $k,s\in\mathbb N$ and suppose that there exists a feasible sequence $(\ell_0,\ldots,\ell_s)$.
	We consider the set
	\[M=\{(a_0,\ldots,a_s)\in \mathbb R_{\geq0}^{s+1} \mid (a_0,\ldots, a_s) \textup{ is feasible}\}\]
	of all feasible sequences. 
	
	For  $(a_0,\ldots,a_s)\in M$, we claim that $a_t\leq (\alpha_j+1)^{t+1}$ for all $0\leq t\leq s$. We show this via a simple induction over $t$. If $t=0$ there is nothing to show since $a_0=0$. For $t>0$  we obtain
	\begin{align*}
		a_t&\leq\alpha_j(a_{t-1}+1)-\Delta_j\sum_{i=0}^{t-1}a_i\leq \alpha_j(a_{t-1}+1)\leq \alpha_j((\alpha_j+1)^{t}+1)\leq (\alpha_j+1)^{t+1}.
	\end{align*}
	Here the first inequality follows from the feasibility of the sequence. 
	As a consequence we see that $M$ is a bounded set. 
	Furthermore $M$ is also closed since $a_0=0, a_t=k$ are both linear inequalities and \eqref{eq:sum} is a linear inequality for all $1\leq t\leq s$. Thus $M$ is compact.
	
	We consider the function $F\colon M\rightarrow \mathbb R$ with $F(a_0,\ldots, a_s)=\sum_{i=0}^{s}a_i$. Since $F$ is continuous and $M$ is compact and non-empty we know that $F$ attains a minimum on $M$, i.e., there is $(b_0,\ldots, b_s)\in M$ with $F(b_0,\ldots, b_s)\leq F(a_0,\ldots, a_s)$ for all $(a_0,\ldots, a_s)\in M$. 
	We claim that $(b_0,\ldots, b_s)$ satisfies \eqref{eq:sum} with equality for all $u+1\leq t\leq s$, where $u$ is the smallest number such that $b_u\neq 0$.
	Suppose this is not the case and let $u+1\leq t\leq s$ be a number such that 
	\begin{align*}
		\frac{b_{t}+\Delta_j\sum_{i=0}^{t-1}b_i}{b_{t-1}+1}<\alpha_j.
	\end{align*}
	If $b_{t-1}=0$, then $(0,\ldots, 0, b_{t},\ldots, b_s)$ is also feasible and moreover 
	\[F(0,\ldots, 0,b_{t},\ldots, b_s)=\sum_{i=t}^{s}b_i<b_u+ \sum_{i=t}^{s}b_i\leq F(b_0,\ldots, b_s)\]
	in contradiction to $(b_0,\ldots, b_s)$ being a minimum.
	Thus we must have $b_{t-1}>0$ and therefore by continuity there exists an $\epsilon\in (0,b_{t-1})$, such that
	\begin{align*}
		\frac{b_t+\Delta_j(b_{t-1}-\epsilon)+\Delta_j\sum_{i=0}^{t-2}b_i}{b_{t-1}-\epsilon+1}\leq\alpha_j.
	\end{align*}
	Observe that the sequence $(c_0,\ldots, c_s)=(b_0,\ldots, b_{t-2},b_{t-1}-\epsilon,b_t,\ldots, b_s)$ is still feasible. The $t$-th inequality is satisfied by choice of $\epsilon$. All other inequalities are satisfied, since for all $1\leq r\leq s$ with $r\neq t$ we have
	
	\begin{align*}
		\frac{c_r+\Delta_j\sum_{i=0}^{r-1}c_i}{c_{r-1}+1}\leq\frac{b_{r}+\Delta_j\sum_{i=0}^{r-1}b_i}{b_{r-1}+1}\leq \alpha_j.
	\end{align*}
	On the other hand
	\begin{align*}
		F(c_0,\ldots, c_s)=\sum_{i=0}^{s}c_i=-\epsilon + \sum_{i=0}^{s}b_i<F(b_0,\ldots, b_s),
	\end{align*}
	which again stands in contradiction to $(b_0,\ldots, b_s)$ being the minimum. Thus $(b_0,\ldots, b_s)$ is of the desired form. 
	
	\textbf{Part 2:}
	Let $k,s\in\mathbb N$ and $(a_0,\ldots,a_s)\in \mathbb R_{\geq0}^{s+1}$  be a feasible sequence which satisfies (\ref{eq:sum}) for all $u+1\leq t\leq s$ with equality, where $u$ is the smallest number such that $a_u\neq 0$. Thus we know that $a_1=\ldots=a_{u-1}=0$ and $a_u\in(0,\alpha_j]$. Furthermore 
	\[a_{u+1}=\alpha_j(a_u+1)-\Delta_j\sum_{i=0}^{u}a_i=\alpha_j(a_u+1)-\Delta_j a_u\]
	
	and for $u+2\leq t\leq s$ we have 
	\begin{align*}
		a_t
		&=\alpha_j(a_{t-1}+1)-\Delta_j\sum_{i=0}^{t-1}a_i\\
		&=\alpha_j(a_{t-1}+1)-\Delta_j a_{t-1}-\Delta_j\sum_{i=0}^{t-2} a_{i}\\
		&=\alpha_j(a_{t-1}+1)-\Delta_ja_{t-1}-(\alpha_j(a_{t-2}+1)-a_{t-1})\\
		&=\alpha_j(a_{t-1}-a_{t-2})-(\Delta_j-1)a_{t-1}.
	\end{align*}
	Here we use that \eqref{eq:sum} is satisfied with equality for $t$ and $t-1$.
	
	Let \[\Psi=\frac{\alpha_j-\Delta_j+1+\sqrt{(\alpha_j-\Delta_j+1)^2-4\alpha_j}}{2}\] and
	\[\Theta=\frac{\alpha_j-\Delta_j+1-\sqrt{(\alpha_j-\Delta_j+1)^2-4\alpha_j}}{2}\] 
	be the two roots of the polynomial $X^2-(\alpha_j-\Delta_j+1)X+\alpha_j$. We observe later that $\Phi\neq\Theta$. Let $x=\frac{\Theta a_u-a_{u+1}}{\Theta-\Phi}$ and $y=\frac{a_{u+1}-\Phi a_{u}}{\Theta-\Phi}$. 
	
	\textbf{Claim:} It holds that $a_t=\Phi^{t-u}x+\Theta^{t-u}y$ for all $u\leq t\leq s$.
	
	We prove this claim by induction over $t$. 
	For $t=u$ we obtain 
	\[x+y=\frac{\Theta a_u-a_{u+1}+a_{u+1}-\Phi a_u}{\Theta-\Phi}=a_u.\]
	For $t=u+1$ we obtain
	\[\Phi x+\Theta y=\frac{\Phi\Theta a_u-\Phi a_{u+1}+\Theta a_{u+1}-\Theta\Phi a_u}{\Theta-\Phi}=a_{u+1}.\]
	For $t>u+1$ we obtain
	\begin{align*}
		&\Phi^{t-u} x+\Theta^{t-u} y\\
		&=\Phi^{t-u-2}x ((\alpha_j-\Delta_j+1)\Phi-\alpha_j) + \Theta^{t-u-2}y ((\alpha_j-\Delta_j+1)\Theta-\alpha_j)\\
		&=\alpha_j((\Phi^{t-u-1}x+\Theta^{t-u-1}y)-(\Phi^{t-u-2}x+\Theta^{t-u-2}y))-(\Delta_j-1)(\Phi^{t-u-1}x+\Theta^{t-u-1}y)\\
		&=\alpha_j(a_{t-1}-a_{t-2})-(\Delta_j-1)a_{t-1}\\
		&=a_t.
	\end{align*}
	For the first equality we used that $\Phi$ and $\Theta$ are roots of $X^2-(\alpha_j-\Delta_j+1)X+\alpha_j$, i.e., $\Phi^2=(\alpha_j-\Delta_j+1)\Phi-\alpha_j$ and $\Theta^2=(\alpha_j-\Delta_j+1)\Theta-\alpha_j$.
	For the third equality we used the induction hypothesis. This proves the claim.
	
	We argue that if $k$ is large enough, there must be $u\leq t \leq s$ with $a_t<0$ in contradiction to our assumption that $(a_0,\ldots, a_s)$ is feasible. For this we observe that by choice of $\alpha_j$ and $\Delta_j$, we get $(\alpha_j-\Delta_j+1)^2-4\alpha_j<0$ and thus $\Phi$ and $\Theta$ are complex numbers. Furthermore $\Phi$ and $\Theta$ are complex conjugates and so are $x$ and $y$. Thus there exists $r>0$ such that the real part of $\Phi^rx$ and $\Theta^ry$ is negative and thus $\Phi^{r}x+\Theta^{r}y$ is negative, see Figure~\ref{fig:complex}.
	
	\begin{figure}[t]
		\centering
		\includegraphics[scale=0.5]{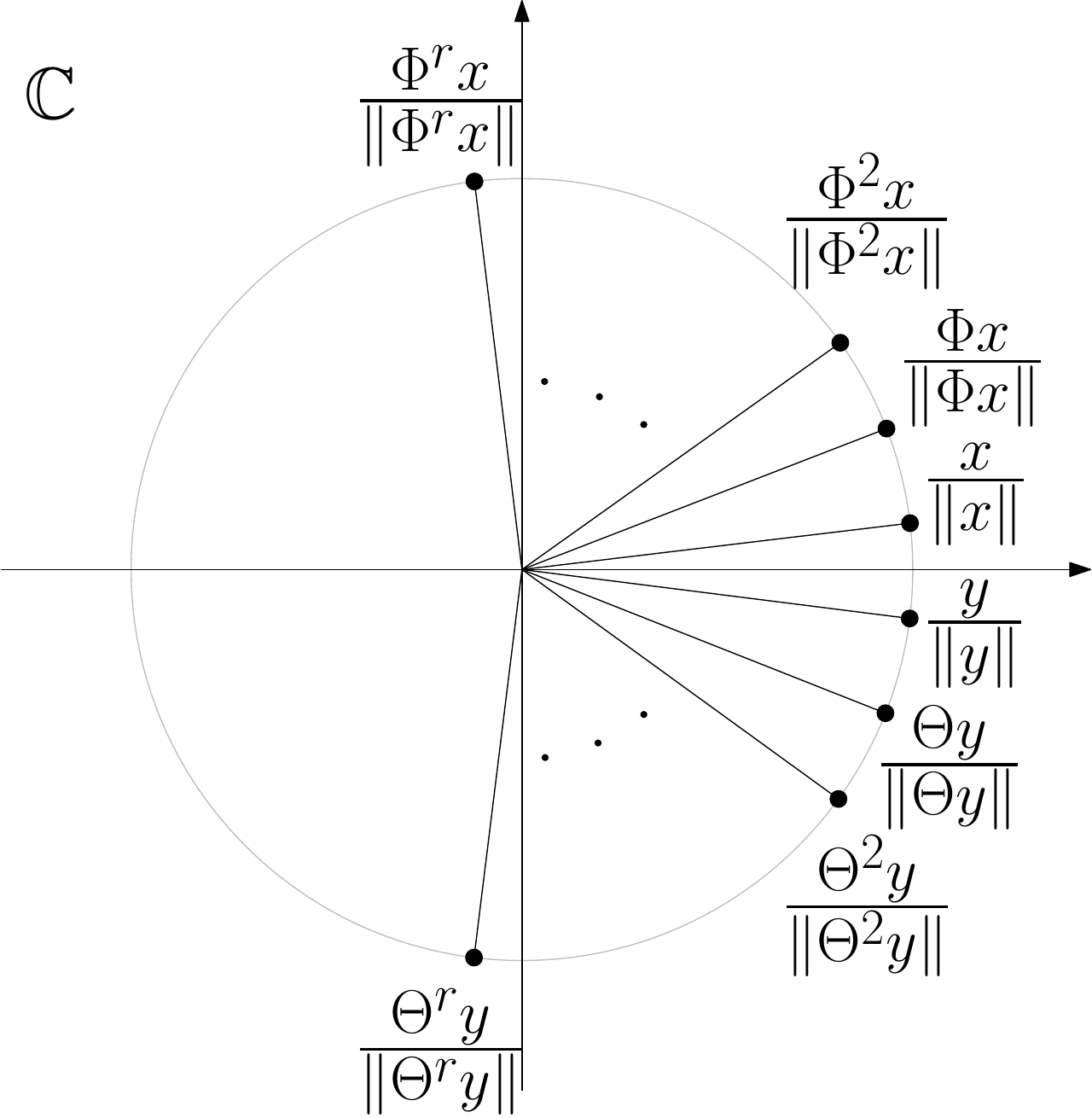}
		\caption{Here we see the normalized numbers on the complex plane.}
		\label{fig:complex}
	\end{figure}
	
	Observe that $a_t\leq (\alpha_j+1)^{t-u+1}$ for $u\leq t\leq s$. One can prove this similar to the bound in Part 1. Thus if $k\geq (\alpha_j+1)^{r}$ we obtain $s\geq r+u$ and thus $a_{r+u}=\Phi^{r}x+\Theta^{r}y$ is negative. Therefore $(a_0,\ldots, a_s)$ is not feasible in contradiction to our assumption.
	
	Let now $k\geq (\alpha_j+1)^{r}$ and suppose there exists $s\in\mathbb N$ and a feasible sequence $(\ell_0,\ldots,\ell_s)$. By the first part we know that there exists a feasible sequence $(a_0,\ldots, a_s)$ which satisfies (\ref{eq:sum}) for all $u+1\leq t\leq s$ with equality, where $u$ is the smallest number such that $a_u\neq 0$. This is in contradiction with the second part, where we prove that for $k\geq (\alpha_j+1)^{r}$ such a sequence cannot exist.
\end{proof}

\section{Counterexample for Mondal's Algorithm}
\label{ap:mondal}
\begin{figure}[t]
	\centering
	\includegraphics[scale=0.7]{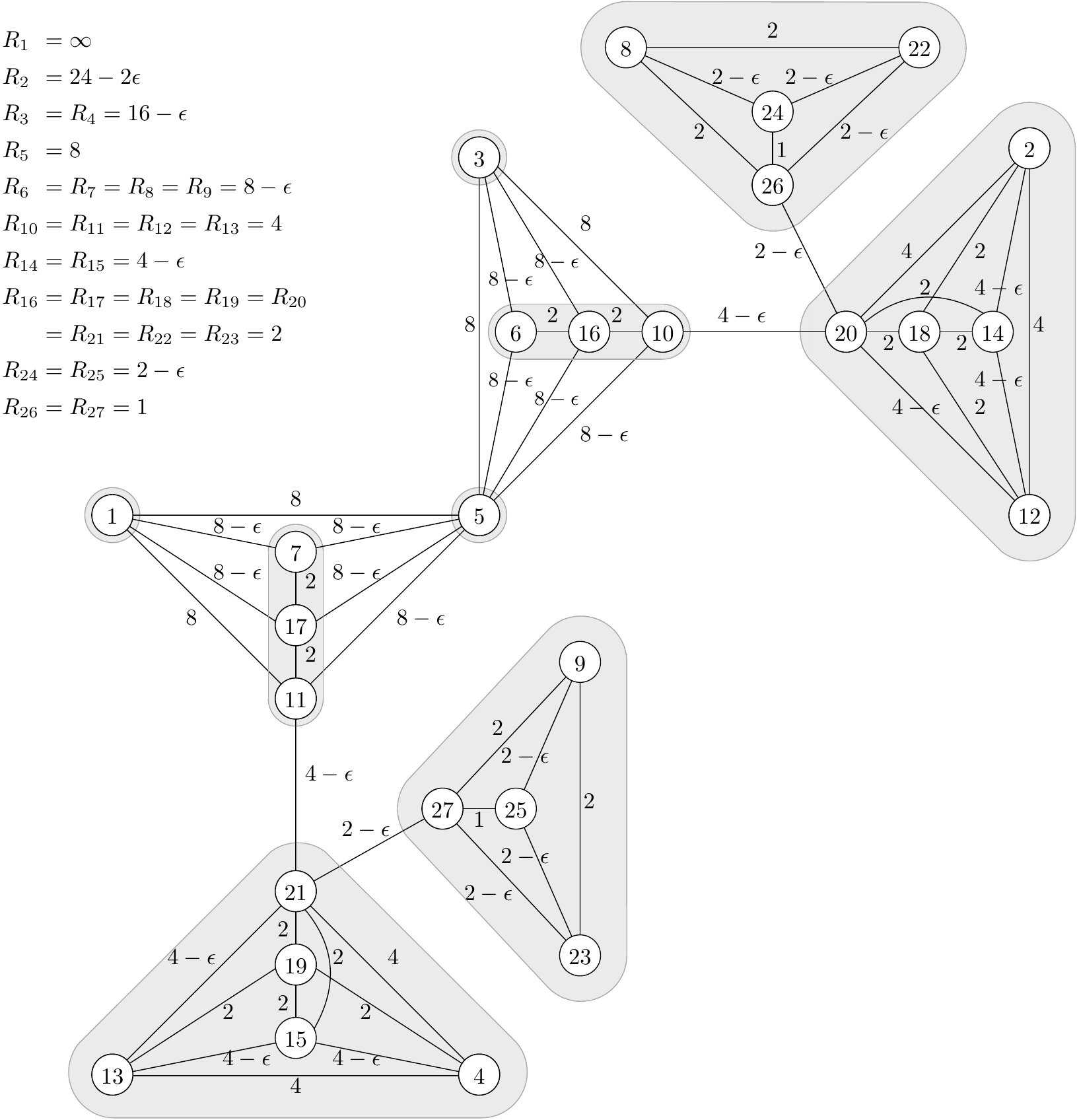}
	\caption{Here we see the clustering instance and the numbering obtained from Gonzales' algorithm as well as the optimal $9$-clustering with radius $2$ depicted in gray. }
	\label{fig:mondal1}
\end{figure}
\begin{figure}[t]
	\centering
	\includegraphics[scale=0.7]{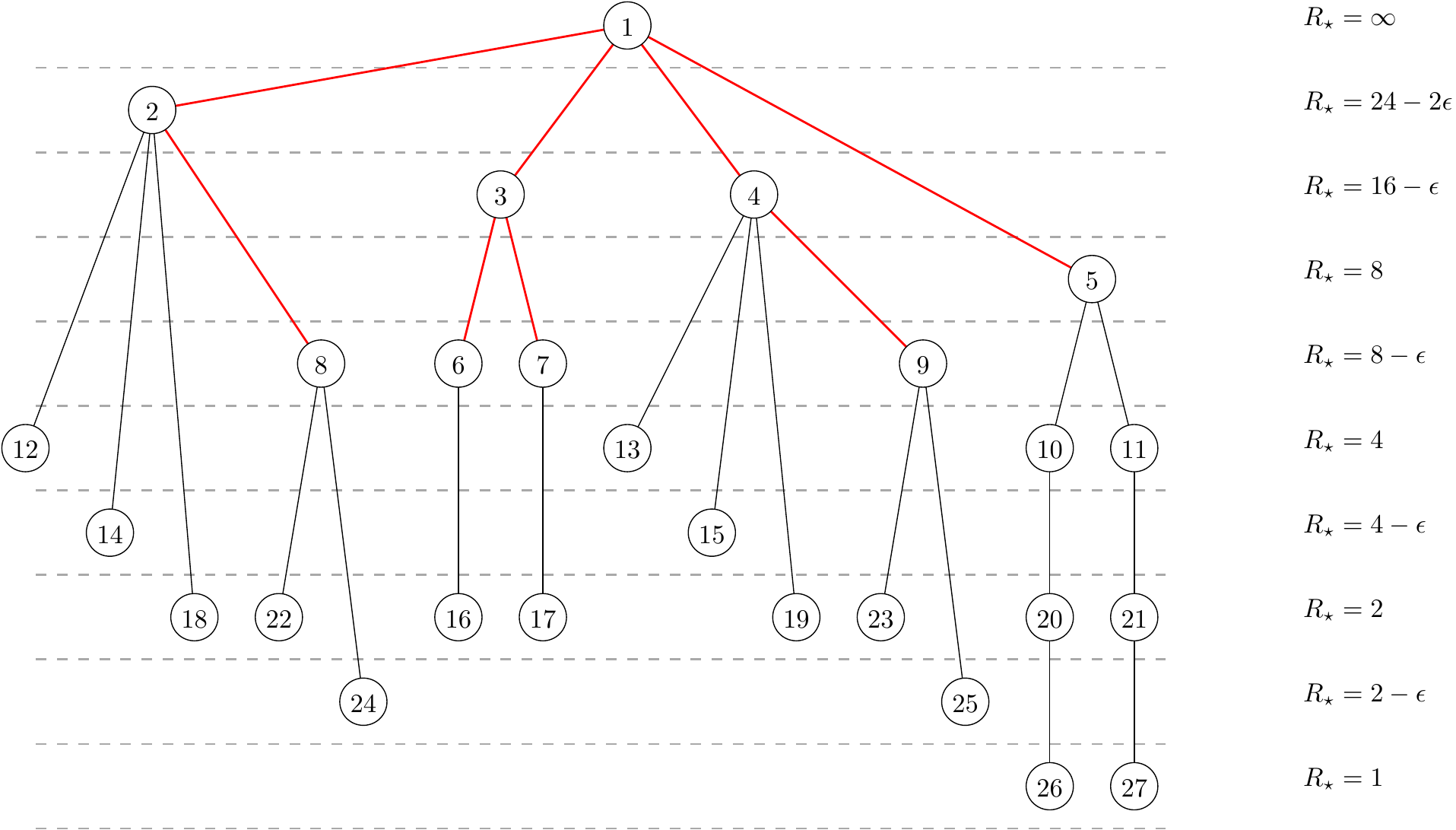}
	\caption{Here we see the final tree. To obtain the $9$-clustering we cut the red edges. The resulting clustering contains the cluster $\{x_5,x_{10},x_{11},x_{20},x_{21},x_{26},x_{27}\}$ of radius $14-3\epsilon$.}
	\label{fig:mondal2}
\end{figure}
The algorithm by Dasgupta and Long~\cite{DasguptaL05} computes a hierarchical clustering which is an 8-approximation with respect to the discrete radius objective and the diameter objective. Mondal's algorithm is a modification of this algorithm and should compute a 6-approximation for the discrete radius objective~\cite[Theorem 3.7]{Mdl18}. We claim that this is not correct and present an example where the approximation factor is $7$. First we give a brief summary of Mondal's algorithm.

Let $(\X,\PP,d)$ be the clustering instance. In the beginning we compute a numbering of the points in $\PP$ by running Gonzales' algorithm~\cite{G85}. The numbering is computed as follows. We pick the first point $x_1\in\PP$ arbitrarily and set $R_{1}=\infty$. For $2\leq k\leq |\PP|$ we set 
\[x_k=\textup{argmax}_{x\in \PP\backslash \{x_1,\ldots, x_{k-1}\}}\min_{1\leq i\leq k-1}d(x,x_i)\]
and $R_k=\min_{1\leq i\leq k-1}d(x_k,x_i)$. In other words the $k$-th point is picked as far as possible from the points $x_1,\ldots, x_{k-1}$ and we denote by $R_k$ the distance of $x_k$ to $x_1,\ldots, x_{k-1}$. 

Based on the $R$-values we define the parent of a point $x\in\PP\backslash\{x_1\}$. Let $N(x)=\textup{argmin}\{d(x,y)\mid y\in\PP, R_x\leq \frac{R_y}{2}\}$ denote the parent of $x$. In other words $N(x)$ is the point nearest to $x$ that satisfies $R_x\leq \frac{R_{N(x)}}{2}$.
Notice that every point in $\PP\backslash\{x_1\}$ has a properly defined parent, as $R_1=\infty$.

We build a tree on $\PP$ as follows. For every point $x\in \PP$ we simply add an edge between $x$ and $N(x)$. The resulting graph is cycle free, since $R_x<R_{N(x)}$ for all $x\in \PP$, and contains $|\PP|-1$ edges. Thus it is indeed a tree.

For any given $1\leq k\leq |\PP|$ we observe that by deleting the edges $\{x_i,N(x_i)\}$ for all $2\leq i\leq k$ the tree decomposes into $k$ connected components with vertex sets $H_k^1,\ldots, H_k^k$. We define the $k$-clustering on $\PP$ to be $\Ht_k=(H_k^1,\ldots, H_k^k)$. Then $\Hr=(\Ht_{|\PP|},\ldots, \Ht_1)$ is a hierarchical clustering of $\PP$. 

We believe that the algorithm by Mondal does not differ significantly from the algorithm by Dasgupta and Long. Since we already know that the analysis of the approximation guarantee of Dasgupta and Long's algorithm is tight~\cite{DM09} the significant improvement on the approximation guarantee seems surprising. We present an example where Mondal's algorithm in fact computes a $7-\epsilon$ approximation for some arbitrarily small $\epsilon>0$, contradicting the claimed approximation guarantee of $6$. We believe that this example can be generalized to prove that the approximation guarantee of Mondal's algorithm is at least $8$. 

Let $\epsilon\in(0,\frac{1}{2})$, Figure~\ref{fig:mondal1} shows a graph with $27$ points which need to be clustered. The metric is given by the shortest path metric in the graph. We perform Mondal's algorithm on this instance under the assumption that we can decide how to break ties, whenever they occur. 

In Figure~\ref{fig:mondal1} we see the numbering of the points which is computed by Gonzales' algorithm as well as all $R$-values. Figure~\ref{fig:mondal2} shows the resulting tree. We obtain the $9$-clustering by cutting all edges $\{x_i,N(x_i)\}$ with $2\leq i\leq 9$. This clustering contains the cluster $\{x_5,x_{10},x_{11},x_{20},x_{21},x_{26},x_{27}\}$, whose radius is $14-3\epsilon$, while the radius of the optimal $9$-clustering is $2$ (see Figure~\ref{fig:mondal1}).

\end{document}